\documentclass[11pt]{article}

\usepackage[colorlinks]{hyperref}
\usepackage{amsmath} 
\usepackage{amsthm} 
\usepackage{amssymb}	
\usepackage{graphicx} 
\usepackage{multicol} 
\usepackage{multirow}
\usepackage{color}
\usepackage[dvips,letterpaper,margin=1in,bottom=1in]{geometry}
\usepackage[capitalize,noabbrev]{cleveref}

\usepackage[utf8]{inputenc}
\usepackage[english]{babel}

\usepackage{diagbox}
\usepackage{mathtools}
\usepackage{adjustbox}

\newtheorem{theorem}{Theorem}[section]

\newtheorem{lemma}[theorem]{Lemma}
\newtheorem{corollary}[theorem]{Corollary}
\newtheorem{proposition}[theorem]{Proposition}
\newtheorem{fact}[theorem]{Fact}

\newtheorem{definition}[theorem]{Definition}

\newtheorem{remark}{Remark}[section]

\newcommand{\braket}[2]{\langle #1 | #2 \rangle}

\DeclarePairedDelimiter\rbra{\lparen}{\rparen}
\DeclarePairedDelimiter\sbra{\lbrack}{\rbrack}
\DeclarePairedDelimiter\cbra{\{}{\}}
\DeclarePairedDelimiter\abs{\lvert}{\rvert}
\DeclarePairedDelimiter\Abs{\lVert}{\rVert}

\DeclarePairedDelimiter\floor{\lfloor}{\rfloor}
\DeclarePairedDelimiter\ket{\lvert}{\rangle}
\DeclarePairedDelimiter\bra{\langle}{\rvert}
\DeclarePairedDelimiter\ave{\langle}{\rangle}

\newcommand{\tr} {\operatorname{tr}}

\newcommand{\spanspace} {\operatorname{span}}

\newcommand{\ketbra}[2]{\ensuremath{\ket{#1}\!\bra{#2}}}

\usepackage{latexsym}
\usepackage{CJK}

\usepackage{enumerate}

\usepackage{algorithm}
\usepackage{algpseudocode}

\usepackage{stmaryrd}
\usepackage{tabularx}
\usepackage{booktabs}
\usepackage{adjustbox}

\newcommand{\footremember}[2]{%
    \footnote{#2}
    \newcounter{#1}
    \setcounter{#1}{\value{footnote}}%
}

\usepackage{tikz}
\usetikzlibrary{quantikz2}

\begin{document}
    \title{Sample-Optimal Quantum Estimators for \\ Pure-State Trace Distance and Fidelity via Samplizer}
\author{Qisheng Wang \footremember{1}{Qisheng Wang is with the School of Computer Science, Shanghai Jiao Tong University, Shanghai 200240, China (e-mail: \url{QishengWang1994@gmail.com}).}
\and Zhicheng Zhang \footremember{2}{Zhicheng Zhang is with the School of Information and Communication Technology, Griffith University, Brisbane, QLD 4111, Australia (e-mail: \url{iszczhang@gmail.com}). Part of the work was done when the author was with the
Centre for Quantum Software and Information, University of Technology Sydney, Australia.}}
        \date{}
        \maketitle

    \begin{abstract}
        We settle the problem of estimating the trace distance and (square root) fidelity between $n$-qubit pure quantum states to within additive error $\varepsilon$, given their independent samples, which was raised as an open question by \hyperlink{cite.Wan24}{Wang (\textit{IEEE Trans.\ Inf.\ Theory} 2024)}.
        This is achieved by a quantum algorithm with \textit{optimal} sample complexity $\Theta(1/\varepsilon^2)$, improving the long-standing folklore with sample complexity $O(1/\varepsilon^4)$.
        At the heart of our algorithm is a samplized phase estimation of the product of two Householder reflections.
        This is realized by an improved \mbox{(multi-)}samplizer for pure states, through which any quantum query algorithm using $Q$ queries to the reflection operator $I - 2\ket{\psi}\!\bra{\psi}$
        can be converted to a $\delta$-close (in the diamond norm distance) quantum sample algorithm using $\Theta(Q^2/\delta)$ samples of the state $\ket{\psi}$. 
        This samplizer for pure states is also shown to be optimal.
    \end{abstract}

    \textbf{Keywords: quantum algorithms, sample complexity, trace distance, fidelity, pure states, lower bounds, samplizer.}

    \newpage
    \tableofcontents
    \newpage

    \section{Introduction}

    Trace distance and fidelity are the two most commonly used measures of the closeness between quantum states in many fundamental problems of quantum computation and quantum information (cf.\ \cite{NC10}), e.g., quantum state discrimination \cite{Che00,BC09,BK15}, certification \cite{BOW19}, and tomography \cite{HHJ+17,OW16}. 
    The trace distance between two mixed quantum states $\rho$ and $\sigma$ is defined by
    \begin{equation}
        \mathrm{T}\rbra{\rho, \sigma} \coloneqq \frac{1}{2}\tr\rbra*{\abs[\big]{\rho - \sigma}}.
    \end{equation}
    The (square root) fidelity between $\rho$ and $\sigma$ is defined by
    \begin{equation}
        \mathrm{F}\rbra{\rho, \sigma} \coloneqq \tr\rbra*{\sqrt{\sqrt{\sigma}\rho\sqrt{\sigma}}}.
    \end{equation}
    To clarify, the squared fidelity $\mathrm{F}^2\rbra{\rho, \sigma} \coloneqq \rbra[\big]{\mathrm{F}\rbra{\rho, \sigma}}^2$ is also a closeness measure commonly used in the literature. 
    As a quantum analog of estimating the closeness (e.g., total variation distance) between probability distributions \cite{VV17}, estimating the trace distance and (square root/squared) fidelity between quantum states turns out to be a basic problem of broad interest. 

\begin{figure} [!htp]
\centering
\begin{quantikz} [row sep = {20pt, between origins}]
    \lstick{$\ket{0}$} & \gate{H} & \ctrl{1} & \gate{H} & \meter{} \\
    \lstick{$\ket{\varphi}$} & & \swap{1} & & \\
    \lstick{$\ket{\psi}$} & & \targX{} & &
\end{quantikz}
\caption{The SWAP test.}
\label{fig:swap}
\end{figure}

    One of the earliest approach to this problem is the SWAP test \cite{BBD+97,BCWdW01}, as shown in \cref{fig:swap}. 
    On input two pure quantum states $\ket{\varphi}$ and $\ket{\psi}$, the measurement outcome of the SWAP test will be $0$ with probability $\frac{1}{2} + \frac{1}{2}\mathrm{F}^2\rbra{\ket{\varphi}, \ket{\psi}}$,\footnote{If we replace the pure quantum states $\ket{\varphi}$ and $\ket{\psi}$ with mixed quantum states $\rho$ and $\sigma$, respectively, then the probability of measurement outcome $0$ will be $\frac{1}{2}+\frac{1}{2}\tr\rbra{\rho\sigma}$ \cite[Proposition 9]{KMY09}.} where 
    \begin{equation}
        \mathrm{F}^2\rbra{\ket{\varphi}, \ket{\psi}} = \abs{\braket{\varphi}{\psi}}^2.
    \end{equation}
    An $\varepsilon$-estimate (in terms of additive error) of $\mathrm{F}^2\rbra{\ket{\varphi}, \ket{\psi}}$ can be obtained with high probability by repeating the SWAP test $O\rbra{1/\varepsilon^2}$ times, which uses $O\rbra{1/\varepsilon^2}$ samples of $\ket{\varphi}$ and $\ket{\psi}$.
    Through this approach, $\varepsilon$-estimates of their trace distance and square root fidelity can be obtained by using $O\rbra{1/\varepsilon^4}$ samples,\footnote{This folklore approach was noted in \cite[Appendix A]{Wan24}.} according to the following identities 
    \begin{equation}
        \mathrm{T}\rbra{\ket{\varphi}, \ket{\psi}} = \sqrt{1 - \mathrm{F}^2\rbra{\ket{\varphi}, \ket{\psi}}}, \qquad \mathrm{F}\rbra{\ket{\varphi}, \ket{\psi}} = \sqrt{\mathrm{F}^2\rbra{\ket{\varphi}, \ket{\psi}}}.
    \end{equation}

Since the discovery of the SWAP test, there has been a series of work on testing/estimating the closeness of quantum states (see \cref{sec:relatedwork} for a comprehensive review).
However, only a few special pure-state cases are known to have optimal algorithms. 
In \cite{ALL22}, it was shown that $\varepsilon$-estimating $\mathrm{F}^2\rbra{\ket{\varphi}, \ket{\psi}}$ requires $\Omega\rbra{1/\varepsilon^2}$ samples of 
$\ket{\varphi}$ and $\ket{\psi}$, thereby implying that the quantum circuit for the SWAP test in \cref{fig:swap} is the best (up to a constant factor) solution to this task that we can hope for; 
equipped with quantum amplitude estimation \cite{BHMT02}, the SWAP test further implies an optimal quantum query algorithm for $\varepsilon$-estimating $\mathrm{F}^2\rbra{\ket{\varphi}, \ket{\psi}}$ by using $O\rbra{1/\varepsilon}$ queries to the state-preparation circuits of $\ket{\varphi}$ and $\ket{\psi}$, where the matching query complexity lower bound is due to \cite{BBC+01,NW99}.
Given these results, one might presume that the SWAP test would also work best for estimating similar quantities such as $\mathrm{T}\rbra{\ket{\varphi}, \ket{\psi}}$ and $\mathrm{F}\rbra{\ket{\varphi}, \ket{\psi}}$. 
A piece of evidence against this is a recent quantum query algorithm proposed in \cite{Wan24} for $\varepsilon$-estimating $\mathrm{T}\rbra{\ket{\varphi}, \ket{\psi}}$ and $\mathrm{F}\rbra{\ket{\varphi}, \ket{\psi}}$ with optimal query complexity $\Theta\rbra{1/\varepsilon}$, which is not based on the SWAP test. 
Nonetheless, the sample complexity of estimating $\mathrm{T}\rbra{\ket{\varphi}, \ket{\psi}}$ and $\mathrm{F}\rbra{\ket{\varphi}, \ket{\psi}}$ remains unsolved. 
The prior best sample complexity upper bound is $O\rbra{1/\varepsilon^4}$ due to the folklore approach just mentioned, while the prior best sample complexity lower bound is only $\Omega\rbra{1/\varepsilon^2}$ as noted in \cite{Wan24}.
These naturally raise a question: 
\begin{align*}
    & \textit{Can we improve the folklore sample complexity $O\rbra{1/\varepsilon^4}$} \\
    & \qquad \qquad \qquad \qquad \qquad \qquad \qquad \qquad \qquad \textit{for estimating $\mathrm{T}\rbra{\ket{\varphi}, \ket{\psi}}$ and $\mathrm{F}\rbra{\ket{\varphi}, \ket{\psi}}$?}
\end{align*}
Surprisingly, we give a positive answer to this question by providing a novel quantum algorithm with \textit{optimal} sample complexity $\Theta\rbra{1/\varepsilon^2}$ for $\varepsilon$-estimating $\mathrm{T}\rbra{\ket{\varphi}, \ket{\psi}}$ and $\mathrm{F}\rbra{\ket{\varphi}, \ket{\psi}}$, matching the lower bound noted in \cite{Wan24}.

We formally state our main result as follows. 

\begin{theorem} [Optimal quantum estimators for pure-state trace distance and square root fidelity, \cref{thm:algo-analysis} restated] \label{thm:main}
    Given sample access to two $n$-qubit pure states $\ket{\varphi}$ and $\ket{\psi}$,  $\mathrm{T}\rbra{\ket{\varphi}, \ket{\psi}}$ and $\mathrm{F}\rbra{\ket{\varphi}, \ket{\psi}}$ can be $\varepsilon$-estimated on a quantum computer with sample complexity $O\rbra{1/\varepsilon^2}$.
\end{theorem}

We summarize the complexity of pure-state trace distance and fidelity estimations in \cref{tab:cmp}. 

\begin{table*}[!htp]
\centering
\caption{Quantum complexity for pure-state trace distance and fidelity estimations.}
\label{tab:cmp}
\adjustbox{max width=\textwidth}{
\begin{tabular}{cccc}
\toprule
Complexity Type             & Trace Distance                                                                                                      & Square Root Fidelity                                                                                                      & Squared Fidelity                                    \\ \midrule
Sample Upper Bound & \begin{tabular}[c]{@{}c@{}}$O(1/\varepsilon^4)$ folklore\\ \textbf{$O(1/\varepsilon^2)$ this work} \end{tabular}                                                                     & \begin{tabular}[c]{@{}c@{}}$O(1/\varepsilon^4)$ folklore\\ \textbf{$O(1/\varepsilon^2)$ this work} \end{tabular}                                                                     & \begin{tabular}[c]{@{}c@{}} $O(1/\varepsilon^2)$ \\ \cite{BCWdW01} \end{tabular}   \\ \midrule
Sample Lower Bound & \begin{tabular}[c]{@{}c@{}} $\Omega(1/\varepsilon^2)$ \\ \cite{Wan24} \end{tabular} & 
\begin{tabular}[c]{@{}c@{}} $\Omega(1/\varepsilon^2)$ \\ \cite{ALL22,Wan24} \end{tabular} & \begin{tabular}[c]{@{}c@{}} $\Omega(1/\varepsilon^2)$ \\ \cite{ALL22} \end{tabular} \\ \midrule
Query Upper Bound  & 
\begin{tabular}[c]{@{}c@{}} $O(1/\varepsilon)$ \\ \cite{Wan24} \end{tabular} & 
\begin{tabular}[c]{@{}c@{}} $O(1/\varepsilon)$ \\ \cite{Wan24} \end{tabular}
& 
\begin{tabular}[c]{@{}c@{}} $O(1/\varepsilon)$ \\ \cite{BCWdW01,BHMT02} \end{tabular}
\\ \midrule
Query Lower Bound  & 
\begin{tabular}[c]{@{}c@{}} $\Omega(1/\varepsilon)$ \\ \cite{Wan24} \end{tabular} & 
\begin{tabular}[c]{@{}c@{}} $\Omega(1/\varepsilon)$ \\ \cite{BBC+01,NW99,Wan24} \end{tabular} & 
\begin{tabular}[c]{@{}c@{}} $\Omega(1/\varepsilon)$ \\ \cite{BBC+01,NW99} \end{tabular} \\ 
 \bottomrule
\end{tabular}
}
\end{table*}

Our quantum algorithm in \cref{thm:main} is quite different from the quantum query algorithm proposed in \cite{Wan24} and those under the framework of the SWAP test. Its key component is a \textit{samplized} phase estimation of the product of two Householder reflections.
To this end, we also extend the \textit{samplizer} (for mixed states) employed in \cite{WZ24} to \textit{(multi-)samplizer for pure states}, which is an algorithmic tool of independent interest.

\paragraph{Organization of this paper.}
We will introduce the techniques in our quantum algorithm and the optimality of our (multi-)samplizer for pure states in \cref{sec:overview}, and give a brief review of related work in \cref{sec:relatedwork}. 
Preliminaries are given in \cref{sec:preliminary}. 
A meta-algorithm using Householder reflections for estimating the trace distance and fidelity will be presented in \cref{sec:est-refl}. 
The samplizer for pure states will be defined and implemented in \cref{sec:samplizer}.
The complete algorithm will be presented in \cref{sec:algo}.
Optimality and lower bounds will be shown in \cref{sec:optimal}.

    \section{Techniques} \label{sec:overview}

    In this section, we introduce the observations and techniques for achieving our results. 
    We first introduce our quantum algorithm in \cref{sec:property-reflection,sec:overview-samplizer}, and then discuss the lower bounds and optimality in \cref{sec:overview-optimal}. 

    \subsection{Properties of the product of Householder reflections} \label{sec:property-reflection}

    The first property of trace distance and square root fidelity we need to notice is that for any two pure quantum states $\ket{\varphi}$ and $\ket{\psi}$, it always holds that 
    \begin{equation} \label{eq:tf}
        \rbra[\big]{\mathrm{T}\rbra{\ket{\varphi}, \ket{\psi}}}^2 + \rbra[\big]{\mathrm{F}\rbra{\ket{\varphi}, \ket{\psi}}}^2 = 1.
    \end{equation}
    \cref{eq:tf} allows us to represent the trace distance and the square root fidelity through trigonometric functions, i.e., one can assume that $\mathrm{T}\rbra{\ket{\varphi}, \ket{\psi}} = \cos\rbra{\gamma}$ and $\mathrm{F}\rbra{\ket{\varphi}, \ket{\psi}} = \sin\rbra{\gamma}$ for some real number $\gamma \in \sbra{0, \frac{\pi}{2}}$.
    This suggests us to estimate $\mathrm{T}\rbra{\ket{\varphi}, \ket{\psi}}$ and $\mathrm{F}\rbra{\ket{\varphi}, \ket{\psi}}$ at the same time by estimating the value of $\gamma$. 
    To this end, our major observation is that $\gamma$ appears to be encoded in the product of the reflection operators about the vectors $\ket{\varphi}$ and $\ket{\psi}$. 

    For convenience, suppose that $\ket{\varphi}$ and $\ket{\psi}$ are two distinguishable pure quantum states, i.e., $\mathrm{T}\rbra{\ket{\varphi}, \ket{\psi}} \neq 0$. 
    Let $R_{\varphi} = I - 2\ketbra{\varphi}{\varphi}$ and $R_{\psi} = I - 2\ketbra{\psi}{\psi}$ be the reflection operators about $\ket{\varphi}$ and $\ket{\psi}$, respectively. 
    Our key observation (\cref{lemma:eigenvector-reflection}) is that the pure quantum state $\ket{\varphi}$ lies in the subspace spanned by two eigenvectors $\ket{\Phi_{\pm}}$ of $R_{\varphi} R_{\psi}$ with eigenvalues $-e^{\mp i2\gamma}$, respectively.
    Specifically, $\ket{\varphi}$ can be represented as
    \begin{equation} \label{eq:def-gamma}
        \ket{\varphi} = \frac{1}{\sqrt{2}} \rbra[\big]{\ket{\Phi_+} + \ket{\Phi_-}}, \text{ where } R_{\varphi} R_{\psi} \ket{\Phi_{\pm}} = -e^{\mp i2\gamma} \ket{\Phi_{\pm}}.
    \end{equation}

    According to the above observation, if we have quantum query access to the two reflection operators $R_{\varphi}$ and $R_{\psi}$, then we can obtain an $\varepsilon$-estimate of the value $\gamma$, denoted by $\tilde \gamma$, via the quantum phase estimation (cf.\ {\cite[Section 5.2]{NC10}}) of $R_{\varphi}R_{\psi}$ on the pure quantum state $\ket{\varphi}$, by using $O\rbra{1/\varepsilon}$ queries to (controlled-)$R_{\varphi}$ and (controlled-)$R_{\psi}$, despite that $\ket{\varphi}$ is not an eigenvector of $R_{\varphi}R_{\psi}$.
    Finally, $\Theta\rbra{\varepsilon}$-estimates of trace distance and square root fidelity are given by $\mathrm{T}\rbra{\ket{\varphi}, \ket{\psi}} \approx \cos\rbra{\tilde \gamma}$ and $\mathrm{F}\rbra{\ket{\varphi}, \ket{\psi}} \approx \sin\rbra{\tilde \gamma}$ (see \cref{lemma:meta} for details).
    For illustration, the quantum circuit for this specific quantum phase estimation is given in \cref{fig:meta}.

    \begin{figure} [!htp]
\centering
\begin{quantikz}
    \lstick{$\ket{0}$} & \gate{H} & & & & \dots \dots & \ctrl{5} & \gate[5]{\mathrm{QFT}_t^\dag} & \meter{} \\
    \lstick{$\vdots$}  \\
    \lstick{$\ket{0}$} & \gate{H} & & & \ctrl{3} & \dots \dots & & & \meter{} \\
    \lstick{$\ket{0}$} & \gate{H} & & \ctrl{2} & & \dots \dots & & & \meter{} \\
    \lstick{$\ket{0}$} & \gate{H} & \ctrl{1} & & & \dots \dots & & & \meter{} \\
    \lstick{$\ket{\varphi}$} & & \gate{\rbra*{R_{\varphi}R_{\psi}}^{2^0}} & \gate{\rbra*{R_{\varphi}R_{\psi}}^{2^1}} & \gate{\rbra*{R_{\varphi}R_{\psi}}^{2^2}} & \dots \dots & \gate{\rbra*{R_{\varphi}R_{\psi}}^{2^{t-1}}} &
\end{quantikz}
\caption{Phase estimation of $R_\varphi R_\psi$ on $\ket{\varphi}$.}
\label{fig:meta}
\end{figure}

    \subsection{Multi-samplizer for pure states} \label{sec:overview-samplizer}

    The \textit{samplizer} employed in \cite{WZ24} is a tool that allows us to convert quantum query algorithms to quantum sample algorithms (the latter use samples of quantum states as input). 
    The algorithm outlined in \cref{fig:meta} can be treated as a quantum query algorithm with the two reflection operators $R_{\varphi}$ and $R_{\psi}$ being the quantum query oracles.
    Now to convert this quantum query algorithm to a quantum sample algorithm using samples of $\ket{\varphi}$ and $\ket{\psi}$,
    we define the notion of \textit{multi-samplizer for pure states}, extending the samplizer in~\cite{WZ24}.
    \begin{definition} [Multi-samplizer for pure states, \cref{def:samplizer-pure} restated] \label{def:samplizer-intro}
        A $k$-samplizer for $n$-qubit pure states, denoted as $\mathsf{Samplize}^{\mathsf{pure}}_{*}\ave{*}$, is a converter such that:
        for any quantum query algorithm $\mathcal{A}^{U_1, U_2, \dots, U_k}$ with $n$-qubit quantum unitary oracles $U_1, U_2, \dots, U_k$ and any $n$-qubit pure states $\ket{\psi_1}, \ket{\psi_2}, \dots, \ket{\psi_k}$, the quantum circuit instance $\mathcal{A}^{R_{\psi_1}, R_{\psi_2}, \dots, R_{\psi_k}}$ can be implemented to precision $\delta$ in diamond norm by using samples of $\ket{\psi_1}, \ket{\psi_2}, \dots, \ket{\psi_k}$, i.e.,
        \begin{equation}
            \Abs*{ \mathsf{Samplize}^{\mathsf{pure}}_{\delta}\ave[\big]{\mathcal{A}^{U_1, U_2, \dots, U_k}}\sbra[\big]{\ket{\psi_1}, \ket{\psi_2}, \dots, \ket{\psi_k}} - \mathcal{A}^{R_{\psi_1}, R_{\psi_2}, \dots, R_{\psi_k}} }_\diamond \leq \delta.
        \end{equation}
    \end{definition}

    Using the notion of multi-samplizer for pure states in \cref{def:samplizer-intro}, we can estimate the trace distance and square root fidelity between pure quantum states by extending the approach in \cref{fig:meta}. 
    Let $\mathsf{QPE}^{U}_{t}$ be the quantum circuit for the phase estimation of $U$ without specifying the input states (see \cref{fig:qpe}), which uses $t$ qubits to store the estimation result of the phase. 
    Then, the quantum circuit in \cref{fig:meta} (without specifying the input state $\ket{0}^{\otimes t} \ket{\varphi}$) can be described by $\mathcal{A}^{R_{\varphi}, R_{\psi}}$, which is a specialization of the quantum query algorithm $\mathcal{A}^{U_1, U_2} \coloneqq \mathsf{QPE}^{U_1 U_2}_{t}$ with $t = \Theta\rbra{\log\rbra{1/\varepsilon}}$. 
    The algorithm is informally outlined as follows (see \cref{algo:td} for the formal version).

    \vspace{5pt}
    \noindent\fbox{
    \parbox{16cm}{
    \vspace{7pt}
    \textbf{Pure-state trace distance and fidelity estimations} (informal): 

    \begin{itemize}
        \item[] \textbf{Input}: Pure quantum states $\ket{\varphi}$ and $\ket{\psi}$.
        \item[] \textbf{Output}: $\varepsilon$-estimates of $\mathrm{T}\rbra{\ket{\varphi}, \ket{\psi}}$ and $\mathrm{F}\rbra{\ket{\varphi}, \ket{\psi}}$.
    \end{itemize}

    \begin{enumerate}
        \item Let $t = \Theta\rbra{\log\rbra{1/\varepsilon}}$ and $\mathcal{A}^{U_1, U_2} \coloneqq \mathsf{QPE}^{U_1U_2}_{t}$.
        \item Prepare the quantum state $\rho \coloneqq \mathsf{Samplize}^{\mathsf{pure}}_{0.01}\ave[\big]{\mathcal{A}^{U_1, U_2}}\sbra[\big]{\ket{\varphi}, \ket{\psi}}\rbra[\big]{\ketbra{0}{0}^{\otimes t} \otimes \ketbra{\varphi}{\varphi}}$.
        \item Obtain an $\varepsilon$-estimate $\tilde \gamma$ of $\gamma$ by measuring the first $t$ qubits of $\rho$. 
        \item Return $\cos\rbra{\tilde \gamma}$ and $\sin\rbra{\tilde \gamma}$ as the estimates of $\mathrm{T}\rbra{\ket{\varphi}, \ket{\psi}}$ and $\mathrm{F}\rbra{\ket{\varphi}, \ket{\psi}}$, respectively. 
    \end{enumerate}
    }
    }
    \vspace{5pt}

    To see the correctness of the above informally stated algorithm, we note that for any $2$-samplizer $\mathsf{Samplize}^{\mathsf{pure}}_{*}\ave{*}$, we have
    \begin{equation}
        \Abs[\big]{ \mathsf{Samplize}^{\mathsf{pure}}_{0.01}\ave[\big]{\mathcal{A}^{U_1, U_2}}\sbra[\big]{\ket{\varphi}, \ket{\psi}} - \mathcal{A}^{R_{\varphi}, R_{\psi}} }_{\diamond} \leq 0.01.
    \end{equation}
    Since an $\varepsilon$-estimate of $\gamma$ can be obtained with high probability, say $\geq 0.99$, by measuring the state $\mathcal{A}^{R_{\varphi}, R_{\psi}}\ket{0}^{\otimes t} \ket{\varphi}$ (in the computational basis), one can obtain an $\varepsilon$-estimate of $\gamma$ with probability $\geq 0.95$ by measuring the following state
    \begin{equation}
        \mathsf{Samplize}^{\mathsf{pure}}_{0.01}\ave[\big]{\mathcal{A}^{U_1, U_2}}\sbra[\big]{\ket{\varphi}, \ket{\psi}}\rbra[\big]{\ketbra{0}{0}^{\otimes t} \otimes \ketbra{\varphi}{\varphi}}.
    \end{equation}
    Let $\tilde \gamma$ be the $\varepsilon$-estimate of $\gamma$ obtained from the above process. 
    Then, we have that $\cos\rbra{\tilde \gamma}$ and $\sin\rbra{\tilde \gamma}$ are $\Theta\rbra{\varepsilon}$-estimates of $\mathrm{T}\rbra{\ket{\varphi}, \ket{\psi}}$ and $\mathrm{F}\rbra{\ket{\varphi}, \ket{\psi}}$, respectively, due to the observation introduced in \cref{sec:property-reflection}. 

    To complete our approach, we provide an efficient implementation of the multi-samplizer for pure states. 

    \begin{theorem} [Implementation of multi-samplizer for pure states, \cref{thm:pure-state-samplizer} restated] \label{thm:samplizer-intro}
        There is an implementation of multi-samplizer for pure states $\mathsf{Samplize}^{\mathsf{pure}}_{*}\ave{*}$ such that if $\mathcal{A}^{U_1, U_2, \dots, U_k}$ uses $Q_j$ queries to $U_j$ for each $1 \leq j \leq k$, then  $\mathsf{Samplize}^{\mathsf{pure}}_{\delta}\ave{\mathcal{A}^{U_1, U_2, \dots, U_k}}\sbra{\ket{\psi_1}, \ket{\psi_2}, \dots, \ket{\psi_k}}$ can be implemented by using 
        \begin{equation}
            O\rbra*{\frac{Q_j}{\delta} \cdot \sum_{i=1}^k Q_i}
        \end{equation}
        samples of $\ket{\psi_j}$ for each $1 \leq j \leq k$. 
    \end{theorem}

    Using the efficient implementation of multi-samplizer for pure states given in \cref{thm:samplizer-intro} (with $k \coloneqq 2$, $Q_1 = Q_2 \coloneqq O\rbra{1/\varepsilon}$ and $\delta \coloneqq 0.01$), the quantum channel $\mathsf{Samplize}^{\mathsf{pure}}_{0.01}\ave[\big]{\mathcal{A}^{U_1, U_2}}\sbra[\big]{\ket{\varphi}, \ket{\psi}}$ can be implemented by using $O\rbra{1/\varepsilon^2}$ samples of $\ket{\varphi}$ and $\ket{\psi}$.
    This yields the upper bounds stated in \cref{thm:main} (see \cref{thm:algo-analysis} for details).

    \begin{remark} \label{remark:samplizer}
        Here, we clarify the difference between the samplizer defined in \cite{WZ24} (see \cref{def:samplizer-general}) and the (multi-)samplizer for pure states defined in \cref{def:samplizer-intro}. 
        \begin{enumerate}
            \item Types of oracles: The samplizer in \cite{WZ24} requires the oracle to be a (scaled) block-encoding of the density operator of a mixed quantum state. By comparison, the samplizer for pure states in \cref{def:samplizer-intro} requires the oracle to be the reflection operator about a pure state. More discussions on the relation between the two types of samplizers are given in \cref{sec:general-samplizer}.
            \item Single/multiple oracles: The samplizer in \cite{WZ24} only allows a single oracle, while the multi-samplizer for pure states in \cref{def:samplizer-intro} allows multiple oracles. 
            \item Sample complexity: The sample complexity of the samplizer in \cite{WZ24} is $\widetilde O\rbra{Q^2/\delta}$ where $\widetilde O\rbra{\cdot}$ suppresses polylogarithmic factors in $Q$ and $1/\delta$, while the sample complexity of the $k$-samplizer for pure states given in \cref{thm:samplizer-intro} is $O\rbra{Q^2/\delta}$ (without polylogarithmic factors) for constant $k$. 
            \item Optimality: The optimality of the samplizer in \cite{WZ24} is shown via a reduction of Hamiltonian simulation. 
            However, the same reduction does not work for the multi-samplizer for pure states in \cref{def:samplizer-intro}. 
            To show the optimality of the multi-samplizer for pure states, our proof uses the algorithm for pure-state trace distance estimation given in \cref{sec:property-reflection} and is also based on the optimal sample lower bound for pure-state trace distance estimation (see \cref{sec:lb-samplizer} for details). 
        \end{enumerate}
    \end{remark}

    \subsection{Optimality} \label{sec:overview-optimal}

    Our quantum algorithm in 
    \cref{thm:main} for pure-state trace distance and square root fidelity estimations is sample-optimal due to the sample complexity lower bound $\Omega\rbra{1/\varepsilon^2}$ noted in \cite{Wan24}. 
    Based on these, we show a matching sample complexity lower bound (up to a constant factor) for the $k$-samplizer for pure states when $k$ is constant. 

    \begin{theorem} [Optimality of the samplizer for pure states, \cref{thm:lb-k-samplizer} restated] \label{thm:lb-samplizer-intro}
        For any $k$-samplizer for pure states, $\mathsf{Samplize}^{\mathsf{pure}}_{*}\ave{*}$, 
        if a quantum query algorithm $\mathcal{A}^{U_1, U_2, \dots, U_k}$ uses $Q_j$ queries to $U_j$ for each $j$, 
        then the implementation of $\mathsf{Samplize}^{\mathsf{pure}}_{\delta}\ave{\mathcal{A}^{U_1, U_2, \dots, U_k}}\sbra{\ket{\psi_1}, \ket{\psi_2}, \dots, \ket{\psi_k}}$ requires $\Omega\rbra{Q_j^2/\delta}$ samples of $\ket{\psi_j}$ for each $j$.
    \end{theorem}
    \begin{proof}[Proof sketch]
    To derive a sample lower bound for the samplizer, we need the meta-algorithm $\mathcal{A}^{U_1, U_2}$ for estimating $\mathrm{T}\rbra{\ket{\varphi}, \ket{\psi}}$, which is based on the observation in \cref{sec:property-reflection} and uses the phase estimation. As will be shown in \cref{corollary:td-fi-est},
    this meta-algorithm $\mathcal{A}^{U_1, U_2}$ makes $O\rbra{1/\varepsilon}$ queries to the reflection operators $U_1 \coloneqq R_{\varphi}$ and $U_2 \coloneqq R_{\psi}$. 
    Suppose that there is an implementation of the \mbox{(multi-)samplizer} for pure states $\mathsf{Samplize}^{\mathsf{pure}}_{*}\ave{*}$ with sample complexity $\mathsf{S}\rbra{Q, \delta}$, meaning that we can implement $\mathsf{Samplize}^{\mathsf{pure}}_{\delta}\ave{\mathcal{B}^{V_1, V_2}}\sbra{\ket{\psi_1}, \ket{\psi_2}}$ by using $\mathsf{S}\rbra{Q, \delta}$ samples of $\ket{\psi_1}$ and $\ket{\psi_2}$ for every quantum query algorithm $\mathcal{B}^{V_1, V_2}$ that makes $Q$ queries to $V_1$ and $V_2$. 
    On the other hand, $\mathsf{Samplize}^{\mathsf{pure}}_{0.01}\ave{\mathcal{A}^{U_1, U_2}}\sbra{\ket{\varphi},\ket{\psi}}$ is an $\varepsilon$-estimator for $\mathrm{T}\rbra{\ket{\varphi}, \ket{\psi}}$ with sample complexity $\mathsf{S}\rbra{\Theta\rbra{1/\varepsilon}, 0.01}$, while any $\varepsilon$-estimator for $\mathrm{T}\rbra{\ket{\varphi}, \ket{\psi}}$ requires sample complexity $\Omega\rbra{1/\varepsilon^2}$, which gives $\mathsf{S}\rbra{\Theta\rbra{1/\varepsilon}, 0.01} \geq \Omega\rbra{1/\varepsilon^2}$. 
    Hence, we can further obtain $\mathsf{S}\rbra{Q, \delta} \geq \Omega\rbra{Q^2/\delta}$ using the properties of $\mathsf{S}\rbra{Q, \delta}$.
    More details will be presented in \cref{thm:lower-bound-samplizer}. 
    \end{proof}

    It is worth mentioning that the sample complexity lower bound for the samplizer for pure states is based on the matching lower bound for pure-state trace distance and square root fidelity estimations, in sharp contrast to the proof for the optimality of the general samplizer in \cite{WZ24} (as discussed in \cref{remark:samplizer}). 
    
    \section{Related Work} \label{sec:relatedwork}

    There are also a few other approaches to estimating trace distance and fidelity in the literature. 
    Except for those based on quantum state tomography \cite{GLF+10,HHJ+17,OW16}, entanglement witnesses \cite{TYKI06,GLGP07,GT09} are a practical technique for estimating the fidelity of certain pure states using few measurements; 
    and direct fidelity estimation \cite{FL11} is designed for pure-state fidelity estimation using only Pauli measurements. 
    In \cite{HKP20}, classical shadows were used to estimate the fidelity between a pure state and a mixed state.
    A distributed quantum algorithm for pure-state squared fidelity estimation was proposed in \cite{ALL22}, which was recently extended to the case with limited quantum computation in \cite{AS24,GHYZ24}. 
    An optimal quantum query algorithm for pure-state trace distance and square root fidelity estimations was proposed in \cite{Wan24}. 
    
    For the case of mixed quantum states, estimating their fidelity are known to have efficient quantum algorithms \cite{WZC+23,WGL+24,GP22,LWWZ24}; especially, in \cite{GP22}, they proposed a sample-efficient quantum algorithm for fidelity estimation with sample complexity $\widetilde O\rbra{r^{5.5}/\varepsilon^{12}}$, where $r$ is the rank of mixed quantum states and $\varepsilon$ is the desired additive error.
    Also, estimating the trace distance are known to have efficient quantum algorithms \cite{WGL+24,WZ24b,LGLW23}; especially, in \cite{WZ24b}, they proposed a sample-efficient quantum algorithm for trace distance estimation with sample complexity $\widetilde O\rbra{r^2/\varepsilon^5}$. 
    Apart from estimation, the certification of $d$-dimensional mixed quantum states is highly related and was studied in \cite{BOW19}, where they provided quantum algorithms that determine whether two mixed quantum states are identical or $\varepsilon$-far with sample complexity $\Theta\rbra{d/\varepsilon^2}$ (w.r.t.\ trace distance) and $\Theta\rbra{d/\varepsilon}$ (w.r.t.\ infidelity), which are optimal due to the lower bounds in \cite{OW21}.

    \section{Preliminaries} \label{sec:preliminary}

    In this section, we introduce different norms used in this paper and list the quantum algorithmic tools essential for designing our quantum algorithms. 

    \subsection{Norms}

    \paragraph{Vector norms.} 
    The $1$-norm (a.k.a.\ Taxicab norm or Manhattan norm) of a $d$-dimensional real-valued vector $v = \rbra{v_0, v_1, \dots, v_{d-1}}^{\mathrm{T}} \in \mathbb{R}^d$ is denoted by
    \begin{equation}
        \Abs{v}_1 = \sum_{j=0}^{d-1} \abs{v_j}.
    \end{equation}
    The total variation distance between two $d$-dimensional probability distributions $p, q \in \mathbb{R}^{d}$ (treated as real-valued vectors) is denoted by 
    \begin{equation}
        d_{\textup{TV}}\rbra{p, q} = \frac{1}{2}\Abs{p - q}_1.
    \end{equation}
    
    A $d$-dimensional (pure) quantum state is described by a complex-valued vector 
    \begin{equation}
        \ket{\psi} = \sum_{j=0}^{d-1} \alpha_j \ket{j} \in \mathbb{C}^d,
    \end{equation}
    where $\alpha_j \in \mathbb{C}$ for each $0 \leq j < d$ is a complex number and $\cbra{\ket{j}}$ is the computational basis. 
    The inner product of two vectors $\ket{\varphi} = \sum_{j=0}^{d-1} \beta_j \ket{j}$ and $\ket{\psi} = \sum_{j=0}^{d-1} \alpha_j \ket{j}$ is denoted by
    \begin{equation}
        \braket{\varphi}{\psi} = \sum_{j=0}^{d-1} \beta_j^* \alpha_j,
    \end{equation}
    where $z^*$ is the complex conjugate of the complex number $z$. 
    The $2$-norm (a.k.a.\ Euclidean norm) of $\ket{\psi}$ is denoted by 
    \begin{equation}
        \Abs{\ket{\psi}} = \sqrt{\braket{\psi}{\psi}}.
    \end{equation}
    
    \paragraph{Matrix norms.} A $d$-dimensional linear operator is described by a $d$-dimensional matrix $A \in \mathbb{C}^{d \times d}$. 
    The matrix norm induced by vector $2$-norm (a.k.a.\ operator norm and spectral norm) is denoted by
    \begin{equation}
        \Abs{A} = \sup_{\ket{\psi} \in \mathbb{C}^d \colon \Abs{\ket{\psi}} \neq 0} \frac{\Abs{A\ket{\psi}}}{\Abs{\ket{\psi}}}.
    \end{equation}
    The Schatten $1$-norm (a.k.a.\ trace norm and nuclear norm) of $A$ is denoted by 
    \begin{equation}
        \Abs{A}_{\tr} = \tr\rbra*{\abs{A}} = \tr\rbra*{\sqrt{A^\dag A}},
    \end{equation}
    where 
    \begin{equation}
        \tr\rbra{A} = \sum_{j=0}^{d-1} \bra{j} A \ket{j}. 
    \end{equation}

    \paragraph{Diamond norm.} A $d$-dimensional superoperator is described by a linear map on $d$-dimensional linear operators, $\Phi \colon \mathbb{C}^{d \times d} \to \mathbb{C}^{d \times d}$. 
    The diamond norm (a.k.a\ completely bounded trace norm) of $\Phi$ is denoted by 
    \begin{equation}
        \Abs{\Phi}_{\diamond} = \sup_{X \in \mathbb{C}^{d^2 \times d^2} \colon \Abs{X}_{\tr} \leq 1} \Abs{\rbra*{\Phi \otimes \mathcal{I}_d} X}_{\tr},
    \end{equation}
    where the superoperator $\mathcal{I}_d$ is the identity operator. 
    It is noteworthy that the diamond distance between two completely positive and trace-preserving maps (a.k.a.\ quantum channels) $\mathcal{E}$ and $\mathcal{F}$ is denoted by $\Abs{\mathcal{E} - \mathcal{F}}_{\diamond}$. 

    \subsection{Quantum phase estimation} \label{sec:QFT}

    Quantum phase estimation \cite{Kit95} is a basic quantum subroutine, which was used in Shor's quantum algorithm for factorization \cite{Sho97}.
    Here, we use the version in \cite{NC10}.

    \begin{theorem} [Quantum phase estimation, {\cite[Section 5.2]{NC10}}] \label{thm:qpe}
        Suppose that $U$ is a unitary operator. 
        There is a quantum circuit $\mathsf{QPE}^{U}_{\varepsilon, \delta}$ using $O\rbra{1/\varepsilon\delta}$ queries to controlled-$U$ that performs the transform
        \begin{equation}
            \mathsf{QPE}^{U}_{\varepsilon, \delta} \colon \ket{0}\ket{\psi} \mapsto \ket{\tilde \lambda}\ket{\psi}
        \end{equation}
        for any eigenvector $\ket{\psi}$ of $U$ with $U\ket{\psi} = e^{i2\pi\lambda}$ and $\lambda \in [0, 1)$,
        where if we measure $\ket{\tilde \lambda}$ in the computational basis (with some classical postprocessing), then with probability at least $1-\varepsilon$ we will obtain a real number $\tilde \lambda \in [0, 1)$ satisfying
        \begin{equation}
            \min\cbra*{\abs{\tilde \lambda - \lambda}, 1 - \abs{\tilde \lambda - \lambda}} < \delta.
        \end{equation}
    \end{theorem}

    An implementation of the quantum phase estimation in \cref{thm:qpe} is also given in {\cite[Section 5.2]{NC10}} (see \cref{fig:qpe}), where the parameter $t$ is chosen to be $t = \Theta\rbra{\log\rbra{1/\varepsilon\delta}}$ and $\mathrm{QFT}_t$ denotes the unitary operator for quantum Fourier transform
    \begin{equation}
        \mathrm{QFT}_t \colon \ket{j} \mapsto \frac{1}{\sqrt{2^t}} \sum_{k=0}^{2^t-1} e^{i2\pi jk/2^t} \ket{k}.
    \end{equation}
    In addition to the queries to the controlled-$U^{2^j}$ gates, the implementation in \cref{fig:qpe} uses $t$ Hadamard gates and an inverse $t$-qubit QFT.
    In \cite{HH00}, it was shown that an (inverse) $t$-qubit QFT can be implemented by using $O\rbra{t \log t}$ two-qubit quantum gates. 

\begin{figure} [!htp]
\centering
\begin{quantikz}
    \lstick{$\ket{0}$} & \gate{H} & & & & \dots \dots & \ctrl{5} & \gate[5]{\mathrm{QFT}_t^\dag} & \meter{} \\
    \lstick{$\vdots$} \\
    \lstick{$\ket{0}$} & \gate{H} & & & \ctrl{3} & \dots \dots & & & \meter{} \\
    \lstick{$\ket{0}$} & \gate{H} & & \ctrl{2} & & \dots \dots & & & \meter{} \\
    \lstick{$\ket{0}$} & \gate{H} & \ctrl{1} & & & \dots \dots & & & \meter{} \\
    \lstick{$\ket{\psi}$} & & \gate{{U}^{2^0}} & \gate{{U}^{2^1}} & \gate{{U}^{2^2}} & \dots \dots & \gate{{U}^{2^{t-1}}} &
\end{quantikz}
\caption{Phase estimation of $U$ on $\ket{\psi}$.}
\label{fig:qpe}
\end{figure}

    \subsection{Density matrix exponentiation}

    Density matrix exponentiation, also called sample-based Hamiltonian simulation, was proposed in \cite{LMR14,KLL+17,GKP+24}, which allows us to approximately implement the unitary operator $e^{-i\rho t}$ by using identical copies of mixed quantum states $\rho$. 
    
    \begin{theorem} [{\cite[Theorem 2]{KLL+17}}] \label{thm:density-matrix-exp}
        Given sample access to a mixed quantum state $\rho$, for every $0 < \delta \leq 1/6$ and $t \geq 6 \pi \delta$, it is necessary and sufficient to use $\Theta\rbra{t^2/\delta}$ samples of $\rho$ to implement a quantum channel that is $\delta$-close to (controlled-)$e^{-i\rho t}$ in diamond norm.
        Moreover, the implementation uses additional $O\rbra{nt^2/\delta}$ one- and two-qubit quantum gates, if $\rho$ is an $n$-qubit quantum state. 
    \end{theorem}

    \section{Estimation with Reflections} \label{sec:est-refl}

    In this section, we study how to estimate the trace distance and square root fidelity between pure quantum states given their Householder reflection operators. 
    We first introduce an important property of the product of two Householder reflections in \cref{sec:prod-refl}.
    Then, using this property, in \cref{sec:meta-algo} we provide a meta-algorithm based on quantum phase estimation. 

    \subsection{Product of Householder reflections} \label{sec:prod-refl}

    For a unit vector $\ket{\psi}$, we use $R_{\psi} = I - 2\ketbra{\psi}{\psi}$ to denote the Householder reflection about $\ket{\psi}$. 
    The following is the key property of the product of Householder reflections that enables our quantum algorithm for pure-state trace distance and fidelity estimations.
    We use $\arg\rbra{a+bi}=\arctan\rbra{b/a}$ to denote the argument of a complex number $a+bi$.

    \begin{lemma} \label{lemma:eigenvector-reflection}
        Suppose that $\ket{\varphi}$ and $\ket{\psi}$ are two unit vectors such that $\abs[\big]{\braket{\varphi}{\psi}} \neq 1$.
        Let
        \begin{equation}
            \ket{\varphi^{\perp}} = \frac{\ket{\psi} - \braket{\varphi}{\psi} \ket{\varphi}}{\Abs*{\ket{\psi} - \braket{\varphi}{\psi} \ket{\varphi}}}.
        \end{equation}
        Let $\theta = \arg\rbra{\braket{\varphi}{\psi}}$ and $\theta^{\perp} = \arg\rbra{\braket{\varphi^{\perp}}{\psi}}$. 
        Then,
        \begin{equation} \label{eq:def-Phi-pm}
            \ket{\Phi_{\pm}} = \frac{1}{\sqrt{2}} \rbra*{ \ket{\varphi} \pm e^{i\rbra*{\theta^\perp - \theta + \frac{\pi}{2}}} \ket{\varphi^{\perp}} }. 
        \end{equation}
        are unit eigenvectors of $R_{\varphi} R_{\psi}$ with eigenvalues $e^{i\rbra*{\pi \mp 2\gamma}}$, where $\gamma = \arcsin\rbra{\abs{\braket{\varphi}{\psi}}}$. 
    \end{lemma}

    \begin{proof}
    Since $\ket{\psi} \in \mathcal{H}_{\varphi} =\spanspace\cbra{\ket{\varphi}, \ket{\varphi^\perp}}$, we can represent $\ket{\psi}$ in terms of $\ket{\varphi}$ and $\ket{\varphi^{\perp}}$ as follows:
    \begin{align}
        \ket{\psi} 
        & = \braket{\varphi}{\psi} \ket{\varphi} + \braket{\varphi^\perp}{\psi} \ket{\varphi^{\perp}} \\
        & = e^{i\theta} \cdot \abs{\braket{\varphi}{\psi}} \cdot \ket{\varphi} + e^{i\theta^\perp} \cdot \abs{\braket{\varphi^\perp}{\psi}} \cdot \ket{\varphi^\perp} \\
        & = e^{i\theta} \sin\rbra{\gamma} \ket{\varphi} + e^{i\theta^\perp} \cos\rbra{\gamma} \ket{\varphi^\perp}. \label{eq:alter-psi}
    \end{align}
    Then, we write out an orthonormal basis of $R_{\psi} \mathcal{H}_{\varphi}$, namely the vectors obtained by applying the reflection $R_{\psi}$ on $\ket{\varphi}$ and $\ket{\varphi^\perp}$:
    \begin{align}
        R_{\psi}\ket{\varphi} & = \ket{\varphi} - 2 \braket{\psi}{\varphi} \ket{\psi}, \label{eq:Rpsi-varphi} \\
        R_{\psi}\ket{\varphi^\perp} & = \ket{\varphi^\perp} - 2 \braket{\psi}{\varphi^\perp} \ket{\psi}. \label{eq:Rpsi-varphi-perp}
    \end{align}
    By \cref{eq:alter-psi,eq:Rpsi-varphi,eq:Rpsi-varphi-perp}, we write the orthonormal basis $\cbra{R_{\varphi}R_{\psi}\ket{\varphi}, R_{\varphi}R_{\psi}\ket{\varphi^\perp}}$ of $R_{\varphi} R_{\psi} \mathcal{H}_{\varphi}$ as follows:
        \begin{align}
            R_{\varphi}R_{\psi}\ket{\varphi} & = R_{\varphi} \rbra*{ \ket{\varphi} - 2 \braket{\psi}{\varphi} \ket{\psi} } \\
            & = -\ket{\varphi} - 2 \braket{\psi}{\varphi} \ket{\psi} + 4 \abs*{\braket{\psi}{\varphi}}^2 \ket{\varphi} \\
            & = - \ket{\varphi} - 2 e^{-i\theta} \sin\rbra{\gamma} \ket{\psi} + 4 \sin^2\rbra{\gamma} \ket{\varphi} \\
            & = \rbra*{2\sin^2\rbra{\gamma} - 1} \ket{\varphi} - 2 e^{i\rbra*{\theta^\perp - \theta}} \sin\rbra{\gamma} \cos\rbra{\gamma} \ket{\varphi^\perp} \\
            & = -\cos\rbra{2\gamma} \ket{\varphi} - e^{i\rbra*{\theta^\perp - \theta}} \sin\rbra{2\gamma} \ket{\varphi^\perp}, \label{eq:RvarphiRpsi-varphi} \\
            R_{\varphi}R_{\psi}\ket{\varphi^\perp} 
            & = R_{\varphi} \rbra*{ \ket{\varphi^\perp} - 2 \braket{\psi}{\varphi^\perp} \ket{\psi} } \\
            & = \ket{\varphi^\perp} - 2 \braket{\psi}{\varphi^\perp} \ket{\psi} + 4 \braket{\varphi}{\psi} \braket{\psi}{\varphi^\perp} \ket{\varphi} \\
            & = \ket{\varphi^\perp} - 2 e^{-i\theta^\perp} \cos\rbra{\gamma} \ket{\psi} + 4 e^{i\rbra*{\theta-\theta^{\perp}}} \sin\rbra{\gamma} \cos\rbra{\gamma} \ket{\varphi} \\
            & = 2e^{i\rbra*{\theta-\theta^\perp}} \sin\rbra{\gamma} \cos\rbra{\gamma} \ket{\varphi} + \rbra*{1 - 2\cos^2\rbra{\gamma}} \ket{\varphi^\perp} \\
            & = e^{i\rbra*{\theta-\theta^\perp}} \sin\rbra{2\gamma} \ket{\varphi} - \cos\rbra{2\gamma} \ket{\varphi^\perp}. \label{eq:RvarphiRpsi-varphi-perp}
        \end{align}
        Combining the above, we can verify the following identity:
        \begin{align}
            R_{\varphi} R_{\psi} \ket{\Phi_{\pm}} 
            & = \frac{1}{\sqrt{2}} R_{\varphi} R_{\psi} \rbra*{ \ket{\varphi} \pm e^{i\rbra*{\theta^\perp - \theta + \frac{\pi}{2}}} \ket{\varphi^\perp} } \\
            & = \frac{1}{\sqrt{2}} \rbra*{ R_{\varphi}R_{\psi}\ket{\varphi} \pm e^{i\rbra*{\theta^\perp - \theta + \frac{\pi}{2}}} R_{\varphi}R_{\psi}\ket{\varphi^{\perp}} } \\
            & = \frac{1}{\sqrt{2}} \sbra*{ \rbra*{-\cos\rbra{2\gamma} \ket{\varphi} - e^{i\rbra*{\theta^\perp - \theta}} \sin\rbra{2\gamma} \ket{\varphi^\perp}} \pm \rbra*{i\sin\rbra{2\gamma}\ket{\varphi} - e^{i\rbra*{\theta^\perp - \theta + \frac{\pi}{2}}} \cos\rbra{2\gamma} \ket{\varphi^\perp}} } \\
            & = \frac{1}{\sqrt{2}} \rbra*{e^{i\rbra*{\pi \mp 2\gamma}}\ket{\varphi} \pm e^{i\rbra*{\theta^\perp - \theta + \frac{3}{2}\pi \mp 2\gamma}} \ket{\varphi^\perp}} \\
            & = e^{i\rbra*{\pi \mp 2\gamma}} \ket{\Phi_{\pm}}, 
        \end{align}
        where the third equality is due to \cref{eq:RvarphiRpsi-varphi,eq:RvarphiRpsi-varphi-perp}.
    \end{proof}

    \subsection{A meta-algorithm} \label{sec:meta-algo}

    We observe the identity 
    \begin{equation}
        \ket{\varphi} = \frac{1}{\sqrt{2}} \rbra[\big]{\ket{\Phi_+} + \ket{\Phi_-}},
    \end{equation}
    where $\ket{\Phi_{\pm}}$ are defined by \cref{eq:def-Phi-pm} and are eigenvectors of $R_{\varphi}R_{\psi}$ as shown in \cref{lemma:eigenvector-reflection}.
    We provide a meta-algorithm by performing the quantum phase estimation of $R_{\varphi}R_{\psi}$ on the quantum state $\ket{\varphi}$, as shown in \cref{fig:meta}.

    \begin{lemma} \label{lemma:meta}
        Suppose that $\varepsilon, \delta \in \rbra{0, 1}$. 
        Let $t = \Theta\rbra{\log\rbra{1/\varepsilon\delta}}$.
        After measuring the first $t$ qubits of $\mathsf{QPE}_{\varepsilon,\delta}^{R_{\varphi}R_{\psi}} \ket{0}^{\otimes t}\ket{\varphi}$ in the computational basis (followed by proper classical postprocessing), with probability at least $1 - \varepsilon$ we will obtain a real number $\tilde \gamma \in [0, 1)$ such that
        \begin{align}
            \abs[\bigg]{\abs*{\sin\rbra*{\pi \tilde \gamma - \frac{\pi}{2}}} - \mathrm{F}\rbra{\ket{\varphi}, \ket{\psi}}} & < \pi\delta, \label{eq:sin-F} \\
            \abs[\bigg]{\abs*{\cos\rbra*{\pi \tilde \gamma - \frac{\pi}{2}}} - \mathrm{T}\rbra{\ket{\varphi}, \ket{\psi}}} & < \pi\delta. \label{eq:cos-T}
        \end{align}
    \end{lemma}
    \begin{proof}
        We first consider the special case that $\ket{\varphi} = \ket{\psi}$ (up to a global phase), i.e., $\abs[\big]{\braket{\varphi}{\psi}} = 1$.
        In this case, $R_{\varphi} R_{\psi} = I$ is the identity operator and $\ket{\varphi}$ is an eigenvector of $R_{\varphi} R_{\psi}$ with eigenvalue $1$. 
        Thus the quantum phase estimation procedure of $\mathsf{QPE}_{\varepsilon,\delta}^{R_{\varphi}R_{\psi}} \ket{0}^{\otimes t}\ket{\varphi}$ will return $\tilde \gamma = 0$ with certainty. 
        By simple calculations, we can verify that $\mathrm{F}\rbra{\ket{\varphi}, \ket{\psi}} = 1 = \abs*{\sin\rbra*{\pi \tilde \gamma - \frac{\pi}{2}}}$ and $\mathrm{T}\rbra{\ket{\varphi}, \ket{\psi}} = 0 = \abs*{\cos\rbra*{\pi \tilde \gamma - \frac{\pi}{2}}}$.

        In the rest of this proof, we consider the general case that $\ket{\varphi} \neq \ket{\psi}$ (up to a global phase), i.e., $\abs[\big]{\braket{\varphi}{\psi}} \neq 1$.
        By \cref{lemma:eigenvector-reflection}, the eigenvalue of $\ket{\Phi_{\pm}}$ with respect to the unitary operator $R_{\varphi}R_{\psi}$ is $e^{i\rbra{\pi\mp2\gamma}}$, where $\gamma = \arcsin\rbra{\abs{\braket{\varphi}{\psi}}}$.
        By \cref{thm:qpe}, we have 
        \begin{align}
            \mathsf{QPE}_{\varepsilon,\delta}^{R_{\varphi}R_{\psi}} \ket{0}^{\otimes t}\ket{\varphi} & = \frac{1}{\sqrt{2}} \rbra*{ \mathsf{QPE}_{\varepsilon,\delta}^{R_{\varphi}R_{\psi}} \ket{0}^{\otimes t}\ket{\Phi_+} + \mathsf{QPE}_{\varepsilon,\delta}^{R_{\varphi}R_{\psi}} \ket{0}^{\otimes t}\ket{\Phi_-} } \\
            & = \frac{1}{\sqrt{2}} \rbra*{\ket{\tilde \gamma_+} \ket{\Phi_+} + \ket{\tilde \gamma_-} \ket{\Phi_-}},
        \end{align}
        where if we measure the first $t$ qubits, then with probability at least $1 - \varepsilon$ we will obtain a real number $\tilde \gamma \in [0, 1)$ satisfying at least one of the following two inequalities:
        \begin{align}
            \min \cbra*{\abs*{\tilde \gamma - \rbra*{\frac{1}{2}-\frac{\gamma}{\pi}}}, 1 - \abs*{\tilde \gamma - \rbra*{\frac{1}{2}-\frac{\gamma}{\pi}}}} < \delta, \label{eq:tildegamma-cond1} \\
            \min \cbra*{\abs*{\tilde \gamma - \rbra*{\frac{1}{2}+\frac{\gamma}{\pi}}}, 1 - \abs*{\tilde \gamma - \rbra*{\frac{1}{2}+\frac{\gamma}{\pi}}}} < \delta. \label{eq:tildegamma-cond2}
        \end{align}
        To see \cref{eq:sin-F}, we note that $\mathrm{F}\rbra{\ket{\varphi}, \ket{\psi}} = \sin\rbra{\gamma}$ and 
        \begin{align}
            \abs[\bigg]{\abs*{\sin\rbra*{\pi \tilde \gamma - \frac{\pi}{2}}} - \sin\rbra{\gamma}}
            & \leq \min\cbra*{ \abs*{\sin\rbra*{\pi \tilde \gamma - \frac{\pi}{2}} - \sin\rbra{\gamma}}, \abs*{\sin\rbra*{\pi \tilde \gamma + \frac{\pi}{2}} - \sin\rbra{\gamma}} } \\
            & \leq \min \Big\{ \abs*{\rbra*{\pi \tilde \gamma - \frac{\pi}{2}} - \gamma}, \pi - \abs*{\rbra*{\pi \tilde \gamma - \frac{\pi}{2}} - \gamma}, \nonumber \\
            & \qquad \qquad \abs*{\rbra*{\pi \tilde \gamma + \frac{\pi}{2}} - \gamma}, \pi - \abs*{{\rbra*{\pi \tilde \gamma + \frac{\pi}{2}} - \gamma}} \Big\} \\
            & < \pi \delta,
        \end{align}
        where the third inequality is because $\tilde \gamma$ satisfies either \cref{eq:tildegamma-cond1} or \cref{eq:tildegamma-cond2}.
        To see \cref{eq:cos-T}, we note that $\mathrm{T}\rbra{\ket{\varphi}, \ket{\psi}} = \cos\rbra{\gamma}$ and similarly we have
        \begin{align}
            \abs[\bigg]{\abs*{\cos\rbra*{\pi \tilde \gamma - \frac{\pi}{2}}} - \cos\rbra{\gamma}}
            & \leq \min\cbra*{ \abs*{\sin\rbra*{\pi \tilde \gamma} - \sin\rbra*{\gamma + \frac{\pi}{2}}}, \abs*{\sin\rbra*{\pi \tilde \gamma + \pi} - \sin\rbra*{\gamma + \frac{\pi}{2}}} } \\
            & \leq \min \Big\{ \abs*{\pi\tilde \gamma - \rbra*{\gamma + \frac{\pi}{2}}}, \pi - \abs*{\pi\tilde \gamma - \rbra*{\gamma + \frac{\pi}{2}}}, \nonumber \\
            & \qquad \qquad \abs*{ \rbra*{\pi \tilde \gamma + \pi} - \rbra*{\gamma + \frac{\pi}{2}} }, \pi - \abs*{ \rbra*{\pi \tilde \gamma + \pi} - \rbra*{\gamma + \frac{\pi}{2}} } \Big\} \\
            & < \pi \delta.
        \end{align}
    \end{proof}

    We restate \cref{lemma:meta} with normalized additive error in the following corollary, which will be used as a subroutine in \cref{sec:algo}. 

    \begin{corollary} \label{corollary:td-fi-est}
        Suppose that $\varepsilon, \delta \in \rbra{0, 1}$. Let $t = \Theta\rbra{\log\rbra{1/\varepsilon\delta}}$.
        We can estimate the trace distance and square root fidelity between $\ket{\varphi}$ and $\ket{\psi}$ to within additive error $\delta$ with probability at least $1  - \varepsilon$ by measuring the first $t$ qubits of
        $\mathsf{QPE}_{\varepsilon,\frac{\delta}{\pi}}^{R_{\varphi}R_{\psi}} \ket{0}^{\otimes t}\ket{\varphi}$ in the computational basis (followed by proper classical postprocessing).
    \end{corollary}
    
    \section{Samplizer for Pure States} \label{sec:samplizer}

    In this section, we define the samplizer for pure states and then provide an efficient implementation for it. 
    Roughly speaking,
    a samplizer converts a quantum query algorithm to a quantum sample algorithm,
    where the query oracle and the sample of quantum states are related.

    \subsection{Definition}

    In our case of pure states, we consider the reflection oracle for pure states, defined as follows. 

    \begin{definition} [Reflection oracle for pure states]
        Let $\ket{\psi}$ be a pure quantum state. 
        The reflection oracle for $\ket{\psi}$ is defined to be $R_\psi = I - 2\ketbra{\psi}{\psi}$. 
    \end{definition}

    We give the definition of (multi-)samplizer for pure states in terms of reflection oracles as follows. 
    
    \begin{definition} [Multi-samplizer for pure states] \label{def:samplizer-pure}
        A $k$-samplizer for $n$-qubit pure states, denoted as $\mathsf{Samplize}^{\mathsf{pure}}_{*}\ave{*}$, is a converter from a quantum circuit family to a quantum channel family such that:
        for any $\delta > 0$, quantum circuit family $\mathcal{A}^{U_1, U_2, \dots, U_k}$ with query access to $n$-qubit unitary operators $U_1, U_2, \dots, U_k$, and $n$-qubit pure states $\ket{\psi_1}, \ket{\psi_2}, \dots, \ket{\psi_k}$, 
        \begin{equation}
            \Abs*{ \mathsf{Samplize}^{\mathsf{pure}}_{\delta}\ave{\mathcal{A}^{U_1, U_2, \dots, U_k}}\sbra{\ket{\psi_1}, \ket{\psi_2}, \dots, \ket{\psi_k}} - \mathcal{A}^{R_{\psi_1}, R_{\psi_2}, \dots, R_{\psi_k}} }_\diamond \leq \delta.
        \end{equation}
        
        The sample complexity of $\mathsf{Samplize}^{\mathsf{pure}}_{*}\ave{*}$ is a $k$-tuple $\rbra{\mathsf{S}_1, \mathsf{S}_2, \dots, \mathsf{S}_k}$ for $\rbra{k+1}$-ary functions $\mathsf{S}_j\rbra{x_1, x_2, \dots, x_k; y}$ such that if $\mathcal{A}^{U_1, U_2, \dots, U_k}$ uses $Q_j$ queries to $U_j$ for each $1 \leq j \leq k$, then $\mathsf{Samplize}^{\mathsf{pure}}_{\delta}\ave{\mathcal{A}^{U_1, U_2, \dots, U_k}}\sbra{\ket{\psi_1}, \ket{\psi_2}, \dots, \ket{\psi_k}}$ uses $\mathsf{S}_j\rbra{Q_1, Q_2, \dots, Q_k; \delta}$ samples of $\ket{\psi_j}$ for each $1 \leq j \leq k$. 
        Especially when $k = 1$, we write $\mathsf{S}\rbra{x, y} \coloneqq \mathsf{S}_1\rbra{x; y}$ for convenience. 

        Similarly, we can also define the (additional) time complexity of $\mathsf{Samplize}^{\mathsf{pure}}_{*}\ave{*}$ as a $\rbra{k+1}$-ary function $\mathsf{T}\rbra{x_1, x_2, \dots, x_k; y}$ such that if $\mathcal{A}^{U_1, U_2, \dots, U_k}$ uses $Q_j$ queries to $U_j$ for each $1 \leq j \leq k$, then $\mathsf{Samplize}^{\mathsf{pure}}_{\delta}\ave{\mathcal{A}^{U_1, U_2, \dots, U_k}}\sbra{\ket{\psi_1}, \ket{\psi_2}, \dots, \ket{\psi_k}}$ can be implemented by using $T_{\mathcal{A}} + \mathsf{T}\rbra{Q_1, Q_2, \dots, Q_k; \delta}$ two-qubit gates, where $T_{\mathcal{A}}$ is the number of two-qubit gates in $\mathcal{A}^{U_1, U_2, \dots, U_k}$ (excluding queries). 
        Especially when $k = 1$, we write $\mathsf{T}\rbra{x, y} \coloneqq \mathsf{T}\rbra{x; y}$ for convenience. 
    \end{definition}

    Our algorithms for trace distance and fidelity estimation employ a $2$-samplizer for pure states, which will be formally described in \cref{sec:algo}. 
    Our definition of samplizer for pure states is inspired by the general samplizer defined in \cite{WZ24}, where the latter simulates quantum query algorithms with query access to block-encoding oracles for quantum states. 
    Here, the block-encoding oracle input model is standard in quantum algorithms, e.g., Hamiltonian simulation \cite{GSLW19} and solving systems of linear equations \cite{CAS+22}.
    It is worth noting that for pure states, the reflection oracle is computationally equivalent to the block-encoding oracle up to a constant factor (cf.\ \cite[Lemma 5.5]{CWZ24}).
    The relationship between the samplizer for pure states and the general samplizer is further discussed in \cref{sec:general-samplizer}. 
    
    \subsection{An efficient approach} \label{sec:approach-samplizer}

    We provide an efficient approach to the multi-samplizer for pure states as follows. 

    \begin{theorem} \label{thm:pure-state-samplizer}
        There is a $k$-samplizer for $n$-qubit pure states, $\mathsf{Samplize}^{\mathsf{pure}}_{*}\ave{*}$, with sample complexity $\rbra{\mathsf{S}_1, \mathsf{S}_2, \dots, \mathsf{S}_k}$, where 
        \begin{equation}
            \mathsf{S}_j\rbra{Q_1, Q_2, \dots, Q_k; \delta} = O\rbra*{\frac{Q_j}{\delta} \cdot \sum_{i=1}^k Q_i}.
        \end{equation} 
        Moreover, if a quantum query algorithm $\mathcal{A}^{U_1, U_2, \dots, U_k}$ uses $Q_j$ queries to $U_j$ for $1 \leq j \leq k$, then the (additional) time complexity (compared to the original $\mathcal{A}^{U_1, U_2, \dots, U_k}$) of  $\mathsf{Samplize}^{\mathsf{pure}}_{\delta}\ave{\mathcal{A}^{U_1, U_2, \dots, U_k}}$ is 
        \begin{equation}
            \mathsf{T}\rbra{Q_1, Q_2, \dots, Q_k; \delta} = O\rbra*{ \frac{n}{\delta} \rbra*{\sum_{j=1}^{k} Q_j}^2 }.
        \end{equation}
    \end{theorem}

    \begin{proof}
        Let $Q \coloneqq \sum_{i=1}^k Q_i$. 
        For each $1 \leq j \leq k$, by \cref{thm:density-matrix-exp} with $t \coloneqq \pi$ and $\delta \coloneqq \delta / Q$, we can use $O\rbra{Q/\delta}$ samples of $\ket{\psi_j}$ and $O\rbra{nQ/\delta}$ one- and two-qubit quantum gates to implement a quantum channel $\mathcal{E}_j$ that is $\rbra{\delta / Q}$-close (in diamond norm) to (the controlled version of) the following unitary transformation
        \begin{equation}
            e^{i \ketbra{\psi_j}{\psi_j} \pi} = I - 2\ketbra{\psi_j}{\psi_j} = R_{\psi_j}.
        \end{equation}
        That is, 
        \begin{equation}
            \Abs*{ \mathcal{E}_j - R_{\psi_j} }_{\diamond} \leq \frac{\delta}{Q}.
        \end{equation}
        Suppose that the quantum circuit $\mathcal{A}^{R_{\psi_1}, R_{\psi_2}, \dots, R_{\psi_k}}$ using queries to $R_{\psi_1}, R_{\psi_2}, \dots, R_{\psi_k}$ is composed of a sequence of unitary operators:
        \begin{equation}
            \mathcal{A}^{R_{\psi_1}, R_{\psi_2}, \dots, R_{\psi_k}} = G_Q \circ \mathcal{O}_Q \circ\dots \circ G_2 \circ \mathcal{O}_2 \circ G_1 \circ \mathcal{O}_1 \circ G_0,
        \end{equation}
        where $G_q$ for each $1 \leq q \leq Q$ can be implemented by one- and two-qubit quantum gates that do not depend on $R_{\psi_1}, R_{\psi_2}, \dots, R_{\psi_k}$, and $\mathcal{O}_q$ for each $1 \leq q \leq Q$ is (the controlled version of) $R_{\psi_{j_q}}$ for some $1 \leq j_q \leq k$. 
        Then, we can implement $\mathsf{Samplize}^{\mathsf{pure}}_{\delta}\ave[\big]{\mathcal{A}^{U_1, U_2, \dots, U_k}}\sbra[\big]{\ket{\psi_1}, \ket{\psi_2}, \dots, \ket{\psi_k}}$ as follows
        \begin{equation} \label{eq:def-samplizer-implement}
            \mathsf{Samplize}^{\mathsf{pure}}_{\delta}\ave[\big]{\mathcal{A}^{U_1, U_2, \dots, U_k}}\sbra[\big]{\ket{\psi_1}, \ket{\psi_2}, \dots, \ket{\psi_k}} \coloneqq G_Q \circ \mathcal{E}_{j_Q} \circ \dots \circ G_2 \circ \mathcal{E}_{j_2} \circ G_1 \circ \mathcal{E}_{j_1} \circ G_0.
        \end{equation}
        As $\Abs{\mathcal{E}_{j_q} - R_{\psi_{j_q}}}_\diamond \leq \delta/Q$ for each $1 \leq q \leq Q$, we conclude that 
        \begin{equation}
            \Abs*{\mathsf{Samplize}^{\mathsf{pure}}_{\delta}\ave[\big]{\mathcal{A}^{U_1, U_2, \dots, U_k}}\sbra[\big]{\ket{\psi_1}, \ket{\psi_2}, \dots, \ket{\psi_k}} - \mathcal{A}^{R_{\psi_1}, R_{\psi_2}, \dots, R_{\psi_k}}}_\diamond \leq \delta.
        \end{equation}
        Therefore, the implementation given in \cref{eq:def-samplizer-implement} is a valid $k$-samplizer for pure states. 

        We analyze the complexity of this implementation as follows. 
        Since there are $Q_j$ queries to $R_{\psi_j}$ among $\mathcal{O}_1, \mathcal{O}_2, \dots, \mathcal{O}_Q$, the number of samples of $\ket{\psi_j}$ used in the implementation defined by \cref{eq:def-samplizer-implement} is 
        \begin{equation}
            Q_j \cdot O\rbra*{\frac{Q}{\delta}} = O\rbra*{\frac{Q_jQ}{\delta}}.
        \end{equation}
        Moreover, the number of one- and two-qubit quantum gates used in this implementation in addition to (the implementation of) $G_1, G_2, \dots, G_Q$ is 
        \begin{equation}
            \sum_{i=1}^k Q_i \cdot O\rbra*{\frac{nQ}{\delta}} = O\rbra*{\frac{nQ^2}{\delta}}.
        \end{equation}
    \end{proof}

    \section{The Algorithm} \label{sec:algo}

    In this section, we present a complete description of our quantum algorithm for pure-state trace distance and fidelity estimations in \cref{algo:td}, which combines the estimation algorithm using reflection oracles in \cref{sec:est-refl} with the multi-samplizer for pure states in \cref{sec:samplizer}. 

    \begin{algorithm}[!htp]
        \caption{Quantum estimator for pure-state trace distance and fidelity estimations.}
        \label{algo:td}
        \begin{algorithmic}[1]
        \Require Sample access to two pure quantum states $\ket{\varphi}$ and $\ket{\psi}$; the desired additive error $\varepsilon \in \rbra{0, 1}$. 

        \Ensure $\varepsilon$-estimates of $\mathrm{T}\rbra{\ket{\varphi}, \ket{\psi}}$ and $\mathrm{F}\rbra{\ket{\varphi}, \ket{\psi}}$ with probability at least $2/3$.

        \State $t \coloneqq \Theta\rbra{\log\rbra{1/\varepsilon}}$.
        
        \State Let $\mathcal{A}^{U_1, U_2} \coloneqq \mathsf{QPE}^{U_1U_2}_{\frac{1}{10}, \frac{\varepsilon}{\pi}}$ denote the quantum circuit for phase estimation of $U_1U_2$.

        \State Prepare the quantum state $\rho \coloneqq \mathsf{Samplize}^{\mathsf{pure}}_{\frac{1}{10}}\ave[\big]{\mathcal{A}^{U_1, U_2}}\sbra[\big]{\ket{\varphi}, \ket{\psi}}\rbra[\big]{\ketbra{0}{0}^{\otimes t} \otimes \ketbra{\varphi}{\varphi}}$.

        \State Let $b_1, b_2, \dots, b_t \in \cbra{0, 1}$ be the outcome of measuring the first $t$ qubits of $\rho$ in the computational basis, and let $\tilde \gamma \coloneqq \sum_{j=1}^t b_j 2^{-j} \in [0, 1)$ be the estimate of the phase. 

        \State \Return $\abs{\cos\rbra{\pi \tilde \gamma - \frac{\pi}{2}}}$ and $\abs{\sin\rbra{\pi \tilde \gamma - \frac{\pi}{2}}}$ as the estimates of $\mathrm{T}\rbra{\ket{\varphi}, \ket{\psi}}$ and $\mathrm{F}\rbra{\ket{\varphi}, \ket{\psi}}$, respectively.
        \end{algorithmic}
    \end{algorithm}

    We state the result of our final algorithm in \cref{thm:algo-analysis}. 

    \begin{theorem} \label{thm:algo-analysis}
        For $\varepsilon \in \rbra{0, 1}$, we can estimate the trace distance and square root fidelity between $n$-qubit pure states $\ket{\varphi}$ and $\ket{\psi}$ to within additive error $\varepsilon$ with probability at least $2/3$ by measuring the first $t$ qubits of 
        \begin{equation} \label{eq:def-final-state}
            \mathsf{Samplize}^{\mathsf{pure}}_{\frac{1}{10}}\ave[\big]{\mathcal{A}^{U_1, U_2}}\sbra[\big]{\ket{\varphi}, \ket{\psi}}\rbra[\big]{\ketbra{0}{0}^{\otimes t} \otimes \ketbra{\varphi}{\varphi}}
        \end{equation}
        in the computational basis (followed by proper classical postprocessing), where $\mathcal{A}^{U_1, U_2} \coloneqq \mathsf{QPE}_{\frac{1}{10},\frac{\varepsilon}{\pi}}^{U_1U_2}$ and $t = \Theta\rbra{\log\rbra{1/\varepsilon}}$.
        Here, the quantum channel $\mathsf{Samplize}^{\mathsf{pure}}_{\frac{1}{10}}\ave[\big]{\mathcal{A}^{U_1, U_2}}\sbra[\big]{\ket{\varphi}, \ket{\psi}}$ can be implemented by using $O\rbra{1/\varepsilon^2}$ samples of $\ket{\varphi}$ and $\ket{\psi}$, and using $O\rbra{n/\varepsilon^2}$ two-qubit gates. 
    \end{theorem}
    \begin{proof}
        The formal description of our algorithm is given in \cref{algo:td}. 
        Now we will prove its correctness and analyze its complexity. 

        \textbf{Correctness}.
        By \cref{corollary:td-fi-est}, we know that the trace distance and square root fidelity between $\ket{\varphi}$ and $\ket{\psi}$ can be obtained with success probability at least $\frac{9}{10}$ by measuring the first $t$ qubits of the state $\mathsf{QPE}^{R_{\varphi}R_{\psi}}_{\frac{1}{10},\frac{\varepsilon}{\pi}} \ket{0}^{\otimes t} \ket{\varphi}$ (followed by proper classical postprocessing). 
        Thus we only have to show that the quantum state defined by \cref{eq:def-final-state} is close to the pure quantum state $\mathsf{QPE}^{R_{\varphi}R_{\psi}}_{\frac{1}{10},\frac{\varepsilon}{\pi}} \ket{0}^{\otimes t} \ket{\varphi}$.
        To this end, by \cref{def:samplizer-pure}, we have
        \begin{equation}
            \Abs*{ \mathsf{Samplize}^{\mathsf{pure}}_{\frac{1}{10}}\ave[\big]{\mathcal{A}^{U_1, U_2}}\sbra[\big]{\ket{\varphi}, \ket{\psi}} -\mathsf{QPE}^{R_{\varphi}R_{\psi}}_{\frac{1}{10},\frac{\varepsilon}{\pi}} }_{\diamond} \leq \frac{1}{10}.
        \end{equation}
        By applying each of them on the quantum state $\Psi \coloneqq \ketbra{0}{0}^{\otimes t} \otimes \ketbra{\varphi}{\varphi}$, we have
        \begin{equation}
            \mathrm{\mathsf{T}}\rbra*{ \mathsf{Samplize}^{\mathsf{pure}}_{\frac{1}{10}}\ave[\big]{\mathcal{A}^{U_1, U_2}}\sbra[\big]{\ket{\varphi}, \ket{\psi}}\rbra[\big]{\Psi}, \mathsf{QPE}^{R_{\varphi}R_{\psi}}_{\frac{1}{10},\frac{\varepsilon}{\pi}} \Psi \rbra*{\mathsf{QPE}^{R_{\varphi}R_{\psi}}_{\frac{1}{10},\frac{\varepsilon}{\pi}}}^\dag } \leq \frac{1}{20}.
        \end{equation}
        Therefore, the probability distribution of the outcomes from measuring (in the computational basis) the state $\mathsf{Samplize}^{\mathsf{pure}}_{\frac{1}{10}}\ave[\big]{\mathcal{A}^{U_1, U_2}}\sbra[\big]{\ket{\varphi}, \ket{\psi}}\rbra[\big]{\Psi}$ is $\frac{1}{10}$-close (in total variation distance) to that from measuring the state $\mathsf{QPE}^{R_{\varphi}R_{\psi}}_{\frac{1}{10},\frac{\varepsilon}{\pi}} \ket{0}^{\otimes t} \ket{\varphi}$ (in the computational basis).
        This implies that the success probability of \cref{algo:td} is at least $\frac{9}{10} - \frac{1}{10} > \frac{2}{3}$. 

        \textbf{Complexity}. 
        By \cref{thm:qpe} with $\varepsilon \coloneqq \frac{1}{10}$ and $\delta \coloneqq \frac{\varepsilon}{\pi}$, the quantum query algorithm $\mathcal{A}^{U_1, U_2} \coloneqq \mathsf{QPE}_{\frac{1}{10},\frac{\varepsilon}{\pi}}^{U_1U_2}$ uses $O\rbra{1/\varepsilon}$ queries to each of $U_1$ and $U_2$, and uses $O\rbra{\log\rbra{1/\varepsilon}\log\log\rbra{1/\varepsilon}}$ two-qubit gates (due to the quantum Fourier transform).
        By the implementation of the multi-samplizer for pure states given in \cref{thm:pure-state-samplizer} with $\delta \coloneqq \frac{1}{10}$ and $Q_1 = Q_2 \coloneqq O\rbra{1/\varepsilon}$, we can implement the quantum channel $\mathsf{Samplize}^{\mathsf{pure}}_{\frac{1}{10}}\ave[\big]{\mathcal{A}^{U_1, U_2}}\sbra[\big]{\ket{\varphi}, \ket{\psi}}$ by using $O\rbra{1/\varepsilon^2}$ samples of each of $\ket{\varphi}$ and $\ket{\psi}$.
        In addition, we also use one sample of $\ket{\varphi}$ as (part of) the input of the quantum channel $\mathsf{Samplize}^{\mathsf{pure}}_{\frac{1}{10}}\ave[\big]{\mathcal{A}^{U_1, U_2}}\sbra[\big]{\ket{\varphi}, \ket{\psi}}$. 
        In total, we use $O\rbra{1/\varepsilon^2}$ samples of $\ket{\varphi}$ and $\ket{\psi}$. 
        Moreover, \cref{algo:td} can be implemented by using 
        \begin{equation}
            O\rbra*{\log\rbra*{\frac{1}{\varepsilon}}\log\log\rbra*{\frac{1}{\varepsilon}}} + O\rbra*{\frac{n}{\varepsilon^2}} = O\rbra*{\frac{n}{\varepsilon^2}}
        \end{equation}
        one- and two-qubit quantum gates, if $\ket{\varphi}$ and $\ket{\psi}$ are $n$-qubit pure quantum states.
    \end{proof}

    \section{Optimality} \label{sec:optimal}

    In this section, we first mention the matching lower bounds for pure-state trace distance and fidelity estimations in \cref{sec:lb-td-fi}, and then show a matching lower bound for the implementation of multi-samplizer for pure states in \cref{sec:lb-samplizer}.

    \subsection{Lower bounds for pure-state trace distance and fidelity estimations} \label{sec:lb-td-fi}

    For completeness, we collect the sample complexity lower bounds for estimating pure-state trace distance and fidelity in the following theorem. 

    \begin{theorem} [Sample complexity lower bounds for pure-state closeness estimation, adapted from \cite{ALL22,Wan24}] \label{thm:td-lb}
        Given sample access to two single-qubit pure quantum states $\ket{\varphi}$ and $\ket{\psi}$, it is necessary to use $\Omega\rbra{1/\varepsilon^2}$ samples of them to $\varepsilon$-estimate 
        \begin{enumerate}[(1)]
            \item the trace distance $\mathrm{T}\rbra{\ket{\varphi}, \ket{\psi}}$ for $\varepsilon \in \rbra{0, 1}$, \label{item:thm-sample-lb-1}
            \item the square root fidelity $\mathrm{F}\rbra{\ket{\varphi}, \ket{\psi}}$ for $\varepsilon \in \rbra{0, 1/4}$. \label{item:thm-sample-lb-3}
            \item the squared fidelity $\mathrm{F}^2\rbra{\ket{\varphi}, \ket{\psi}}$ for $\varepsilon \in \rbra{0, 1/2}$, \label{item:thm-sample-lb-2}
        \end{enumerate}
    \end{theorem}
    \begin{proof}
        Items \ref{item:thm-sample-lb-3} and \ref{item:thm-sample-lb-2} are directly implied by \cite[Lemma 13 in the full version]{ALL22}. 
        It remains to prove Item \ref{item:thm-sample-lb-1}. 

        The sample complexity lower bound for pure-state trace distance estimation was already shown in \cite[Theorem B.2]{Wan24} for $\varepsilon \in \rbra{0, 1/2}$. 
        Here, we strengthen the lower bound to encompass a more general case that $\varepsilon \in \rbra{0, 1}$ and both the states $\ket{\varphi}$ and $\ket{\psi}$ are single-qubit, with a simple proof. 
        To this end, we denote
        \begin{equation} \label{eq:def-psix}
            \ket{\psi_x} = \sqrt{1 - x^2} \ket{0} + x \ket{1}
        \end{equation}
        for $x \in \sbra{0, 1}$.
        Let $\varepsilon \in \rbra{0, 1}$.
        Consider the following quantum hypothesis testing problem: 
        \begin{itemize}
            \item Given an unknown pure quantum state $\ket{\varphi}$, determine whether $\ket{\varphi}$ is $\ket{\psi_0}$ or $\ket{\psi_\varepsilon}$, under the promise that $\ket{\varphi}$ is in either case with equal probability. 
        \end{itemize}
        The trace distance between $\ket{\psi_0}$ and $\ket{\psi_\varepsilon}$ is
        \begin{equation}
            \mathrm{T}\rbra*{\ket{\psi_0}, \ket{\psi_\varepsilon}} = \sqrt{1 - \abs*{\braket{\psi_0}{\psi_\varepsilon}}^2} = \varepsilon.
        \end{equation}
        By the Holevo-Helstrom bound for quantum state discrimination \cite{Hol73,Hel67}, we know that the above quantum hypothesis testing problem requires sample complexity $\Omega\rbra{1/\varepsilon^2}$. 
        On the other hand, any estimator for pure-state trace distance with additive error $\varepsilon$ can be used to solve this task, thereby also requiring sample complexity $\Omega\rbra{1/\varepsilon^2}$.
    \end{proof}

    \subsection{Lower bounds for (multi-)samplizers for pure states} \label{sec:lb-samplizer}

    It is trivial that for any $k \geq 1$, any $\rbra{k+1}$-samplizer for pure states can be used to implement a $k$-samplizer for pure states with the same sample complexity (by discarding the last pure state). 
    To show the optimality of the multi-samplizer for pure states given in \cref{thm:pure-state-samplizer}, it is sufficient to show the optimality of the $1$-samplizer, given as follows. 

    \begin{theorem} \label{thm:lower-bound-samplizer}
        For sufficiently large $Q$ and every $\delta \in \rbra{0, 1/9}$, any $1$-samplizer for $n$-qubit pure states requires sample complexity $\mathsf{S}\rbra{Q, \delta} = \Omega\rbra{Q^2/\delta}$. 
    \end{theorem}

    Although the lower bound given in \cref{thm:lower-bound-samplizer} for the $1$-samplizer for pure states is of the same order as the lower bound $\Omega\rbra{Q^2/\delta}$ for the general samplizer given in \cite[Theorem III.4]{WZ24}, we emphasize that prior lower bounds do \textit{not} imply our lower bound in \cref{thm:lower-bound-samplizer}. 
    We discuss related lower bounds as follows. 
    \begin{enumerate}
        \item The lower bound $\Omega\rbra{Q^2/\delta}$ for general samplizer given in \cite[Theorem III.4]{WZ24} is by reducing from the sample-based Hamiltonian simulation \cite{KLL+17} via the optimal quantum query algorithm for Hamiltonian simulation \cite{GSLW19} for time $t = \Theta\rbra{Q}$.
        The sample-based Hamiltonian simulation is a task for implementing $e^{-i\rho t}$ with precision $\delta$ using samples of $\rho$. A lower bound $O\rbra{t^2/\delta}$ on the sample complexity of sample-based Hamiltonian simulation is given in \cite[Theorem 2]{KLL+17}, which is further reduced from distinguishing two mixed quantum states $\rho\rbra{\frac 1 2}$ and $\rho\rbra{\frac 1 2+\Theta\rbra{\frac 1 t}}$, where $\rho\rbra{x} = x\ketbra{0}{0} + \rbra{1-x}\ketbra{1}{1}$. 
        However, this reduction does not directly apply to the case of pure states.
        \item One may wonder what happens if we consider the sample-based Hamiltonian simulation for pure states. Note that the unitary operator $e^{-i\ketbra{\psi}{\psi}t}$ is periodic with respect to $t$ as $e^{-i\ketbra{\psi}{\psi}2\pi} = I$. 
        So the above idea will result in a lower bound of $\mathsf{S}\rbra{Q, \delta} = \Omega\rbra{1/\delta}$ for \mbox{$1$-samplizer} by taking $Q = \Theta\rbra{t} = \Theta\rbra{1}$, due to the lower bound $\Omega\rbra{1/\delta}$ given in \cite[Theorem 4]{KLL+17} for implementing $e^{-i\ketbra{\psi}{\psi}t}$ for constant $t$ with precision $\delta$.
    \end{enumerate}

    Interestingly, our proof of \cref{thm:lower-bound-samplizer} is based on a reduction from the pure-state trace distance estimation (whose lower bounds are previously presented in \cref{sec:lb-td-fi}) instead of the sample-based Hamiltonian simulation~\cite{KLL+17}.

    \begin{proof} [Proof of \cref{thm:lower-bound-samplizer}]
        Let $\mathcal{A}^{U_1, U_2}_\varepsilon \coloneqq \mathsf{QPE}_{\frac{1}{10},\frac{\varepsilon}{\pi}}^{U_1U_2}$.
        Then, by \cref{corollary:td-fi-est}, measuring the quantum state $\mathcal{A}^{R_\varphi, R_\psi}_\varepsilon \ket{0}^{\otimes t}\ket{\varphi}$ can estimate the trace distance between two pure quantum states $\ket{\varphi}$ and $\ket{\psi}$ to within additive error $\varepsilon$ with query complexity $O\rbra{1/\varepsilon}$, where $R_{\varphi}$ and $R_{\psi}$ are the reflection oracles for $\ket{\varphi}$ and $\ket{\psi}$, respectively.
        Consider the two pure quantum states $\ket{\psi_0} \coloneqq \ket{0}$ and $\ket{\psi_\varepsilon} \coloneqq \sqrt{1 - \varepsilon^2}\ket{0} + \varepsilon\ket{1}$ employed in \cref{thm:td-lb}, where $\varepsilon \in \rbra{0, 1}$.
        Note that the reflection oracle $R_{\psi_0} = I - 2\ketbra{0}{0}$ for $\ket{\psi_0}$ is constant (with respect to $\varepsilon$). 
        Then, measuring $\mathcal{A}^{R_{\psi_0}, R_{\psi_\varepsilon}}_\varepsilon \ket{0}$ (note that $\ket{\psi_0} = \ket{0}$) can estimate the trace distance between $\ket{\psi_0}$ and $\ket{\psi_\varepsilon}$ to within additive error $\varepsilon$ with probability at least $2/3$, using $O\rbra{1/\varepsilon}$ queries to $R_{\psi_\varepsilon}$. 
        Therefore, using the quantum query algorithm $\mathcal{A}^{U_1, U_2}_\varepsilon$ for pure-state trace distance estimation, we can obtain a quantum query algorithm $\mathcal{T}^{U}$ for distinguishing the two reflection operators $R_{\psi_0}$ and $R_{\psi_\varepsilon}$ with probability at least $2/3$, using $O\rbra{1/\varepsilon}$ queries to $U$, promised that either $U = R_{\psi_0}$ or $R_{\psi_\varepsilon}$. 
        Formally, the probability that $\mathcal{T}^U$ accepts is given by $\Abs{\Pi \mathcal{T}^U \ket{0}}^2$, where $\Pi = \ketbra{0}{0} \otimes I$ measures the first qubits.
        Then, 
        \begin{itemize}
            \item $\mathcal{T}^{R_{\psi_0}}$ accepts with probability at least $2/3$; 
            \item $\mathcal{T}^{R_{\psi_\varepsilon}}$ accepts with probability at most $1/3$.
        \end{itemize}

        Let $\mathsf{Samplize}^{\mathsf{pure}}_*\ave{*}$ be any $1$-samplizer with sample complexity $\mathsf{S}\rbra{Q, \delta}$. 
        Then, $\mathsf{Samplize}^{\mathsf{pure}}_{\frac{1}{9}} \ave{\mathcal{T}^{U}}$ is a quantum channel family with sample complexity $\mathsf{S}\rbra{\Theta\rbra{1/\varepsilon}, 1/9}$ such that 
        \begin{itemize}
            \item $\mathsf{Samplize}^{\mathsf{pure}}_{\frac{1}{9}} \ave[\big]{\mathcal{T}^{U}}\sbra[\big]{\ket{\psi_0}}\rbra[\big]{\ketbra{0}{0}}$ accepts with probability 
            \begin{equation}
            \tr\rbra*{\Pi \cdot \mathsf{Samplize}^{\mathsf{pure}}_{\frac{1}{9}} \ave[\big]{\mathcal{T}^{U}}\sbra[\big]{\ket{\psi_0}}\rbra[\big]{\ketbra{0}{0}}} \geq \frac{2}{3} - \frac{1}{9} = \frac{5}{9},
            \end{equation}
            \item $\mathsf{Samplize}^{\mathsf{pure}}_{\frac{1}{9}} \ave[\big]{\mathcal{T}^{U}}\sbra[\big]{\ket{\psi_\varepsilon}}\rbra[\big]{\ketbra{0}{0}}$ accepts with probability
            \begin{equation}
            \tr\rbra*{\Pi \cdot \mathsf{Samplize}^{\mathsf{pure}}_{\frac{1}{9}} \ave[\big]{\mathcal{T}^{U}}\sbra[\big]{\ket{\psi_\varepsilon}}\rbra[\big]{\ketbra{0}{0}}} \leq \frac{1}{3} + \frac{1}{9} = \frac{4}{9}.
            \end{equation}
        \end{itemize}
        Therefore, if $\ket{\varphi}$ is promised to be either $\ket{\psi_0}$ or $\ket{\psi_\varepsilon}$, then, by a constant number (say, $9$) of repetitions of $\mathsf{Samplize}^{\mathsf{pure}}_{\frac{1}{9}} \ave{\mathcal{T}^{U}}$ on input $\ket{\varphi}$, we can determine with probability at least $2/3$ which is the case (by majority voting).
        This means that we can distinguish $\ket{\psi_0}$ and $\ket{\psi_\varepsilon}$ with probability at least $2/3$ with sample complexity $9 \cdot \mathsf{S}\rbra{\Theta\rbra{1/\varepsilon}, 1/9}$. 
        On the other hand, using the same arguments as in \cref{thm:td-lb}, we know that distinguishing $\ket{\psi_0}$ and $\ket{\psi_\varepsilon}$ requires sample complexity $\Omega\rbra{1/\varepsilon^2}$.
        These together imply that $9 \cdot \mathsf{S}\rbra{\Theta\rbra{1/\varepsilon}, 1/9} \geq \Omega\rbra{1/\varepsilon^2}$.
        Since the choice of $\varepsilon \in (0, 1)$ is arbitrary, by letting $Q = \Theta\rbra{1/\varepsilon}$, we further have the relation
        \begin{equation} \label{eq:S-constant-delta}
            \mathsf{S}\rbra*{Q, \frac{1}{9}} \geq c \cdot Q^2
        \end{equation}
        for sufficiently large $Q$, where $c > 0$ is a universal constant.
        This is the desired inequality for the case of constant $\delta$. 

        To show the inequality for arbitrarily small $\delta$, we note that $\mathsf{S}\rbra{Q, \delta}$ satisfies:
        \begin{equation}
            \label{eq:pro-S}
            \mathsf{S}\rbra*{Q_1+Q_2,\delta_1+\delta_2}\leq \mathsf{S}\rbra*{Q_1,\delta_1}+\mathsf{S}\rbra*{Q_2,\delta_2}
        \end{equation}
        for any integers $Q_1,Q_2\geq 0$ and real numbers $\delta_0,\delta_1\in (0,1)$.
        This is because: to samplize any quantum query algorithm with $Q_1+Q_2$ queries to precision $\delta_1+\delta_2$, one can always first samplize the first part of $Q_1$ queries to precision $\delta_1$ and then the second part of $Q_2$ queries to precision $\delta_2$. 
        Consequently, we further have the following properties:
        \begin{enumerate}
            \item \label{enum:property-1} $\mathsf{S}\rbra{mQ, m\delta} \leq m \cdot \mathsf{S}\rbra{Q, \delta}$ for any integers $m, Q \geq 1$ and real number $\delta > 0$.
            \item \label{enum:property-2} $\mathsf{S}\rbra{Q, \delta_1} \geq \mathsf{S}\rbra{Q, \delta_2}$ for any integer $Q \geq 1$ and real numbers $0 < \delta_1 < \delta_2$.
        \end{enumerate}
        For sufficiently large $Q$ and for every $0 < \delta < 1/9$, by taking $m \coloneqq \floor*{\frac{1}{9\delta}} \geq 1$ (which gives $0 < m\delta \leq 1/9$), we have
        \begin{align}
            \mathsf{S}\rbra{Q, \delta} 
            & \geq \frac{1}{m} \cdot \mathsf{S}\rbra*{mQ, m\delta} \\
            & \geq \frac{1}{m} \cdot \mathsf{S}\rbra*{mQ, \frac{1}{9}} \\
            & \geq \frac{1}{m} \cdot c \cdot \rbra*{mQ}^2 \\
            & = c \cdot \floor*{\frac{1}{9\delta}} \cdot Q^2 \\
            & \geq \Omega\rbra*{\frac{Q^2}{\delta}},
        \end{align}
        where the first inequality is due to property \ref{enum:property-1}, the second inequality is due to property \ref{enum:property-2}, and the third inequality is by \cref{eq:S-constant-delta}.

        To complete the proof for general $n \geq 2$, the hard instance can be taken as $\ket{\Psi_0} = \ket{\psi_0} \otimes \ket{0}^{\otimes \rbra{n-1}}$ and $\ket{\Psi_\varepsilon} = \ket{\psi_\varepsilon} \otimes \ket{0}^{\otimes \rbra{n-1}}$. 
    \end{proof}

    As a corollary of \cref{thm:lower-bound-samplizer}, we have a matching sample lower bound for the implementation of $k$-samplizer for pure states for constant $k$. 

    \begin{theorem} \label{thm:lb-k-samplizer}
        For sufficient large $Q_j$ for each $1 \leq j \leq k$ and every $\delta \in \rbra{0, 1/9}$, any $k$-samplizer for pure states requires sample complexity 
        \begin{equation}
            \mathsf{S}_j\rbra{Q_1, Q_2, \dots, Q_k; \delta} = \Omega\rbra*{\frac{Q_j^2}{\delta}}
        \end{equation}
        for each $1 \leq j \leq k$. 
    \end{theorem}

    \section{Conclusion}

    In this paper, we proposed a quantum algorithm for estimating the trace distance and square root fidelity between pure quantum states with \textit{optimal} sample complexity, quadratically improving the long-standing folklore approach, answering the question raised in \cite{Wan24}, and thereby completing the complexity picture for pure-state closeness estimation. 
    Technically, our quantum algorithm requires new observations (the properties of the product of Householder reflections) and tools (the samplizer for pure states). 

    We therefore raise the following questions for future research. 

    \begin{enumerate}
        \item Can we close the gap between the upper and lower bounds on the quantum query/sample complexities of estimating the quantities of mixed quantum states such as trace distance and fidelity?
        \item Can we find more applications of the samplizer for pure states?
        \item After the distributed quantum algorithm for pure-state squared fidelity estimation in \cite{ALL22}, the trade-off between the sample complexity and quantum communication was recently settled in \cite{AS24,GHYZ24}. 
        A further question is: what is the sample complexity of estimating the square root fidelity and trace distance between pure states with limited quantum communication? \label{ques:3}
        \item Following Question \ref{ques:3}, can we implement a samplizer (for pure states) with distributed quantum computation or with limited quantum communication?
    \end{enumerate}

    \section*{Acknowledgment}

    The work of Qisheng Wang was supported by the Startup Funding from Shanghai Jiao Tong University.
    The work of Zhicheng Zhang was  supported in part by the Australian Research Council under Grant DP250102952.

    \addcontentsline{toc}{section}{References}
    
    \bibliographystyle{alphaurl}
    \bibliography{main}

@article{WZ24,
    author = {Wang, Qisheng and Zhang, Zhicheng},
    title = {Time-efficient quantum entropy estimator via samplizer},
    journal = {IEEE Transactions on Information Theory},
    volume = {71},
    number = {12},
    pages = {9569--9599},
    doi = {10.1109/TIT.2025.3576137},
    year = {2025}
}

@misc{CWZ24,
    author = {Chen, Kean and Wang, Qisheng and Zhang, Zhicheng},
    title = {Local test for unitarily invariant properties of bipartite quantum states},
    howpublished = {ArXiv e-prints},
    eprint = {2404.04599},
    year = {2024}
}

@inproceedings{GSLW19,
    author = {Gily\'{e}n, Andr\'{a}s and Su, Yuan and Low, Guang Hao and Wiebe, Nathan},
    title = {Quantum singular value transformation and beyond: exponential improvements for quantum matrix arithmetics},
    booktitle = {Proceedings of the 51st Annual ACM SIGACT Symposium on Theory of Computing},
    pages = {193-204},
    doi = {10.1145/3313276.3316366},
    year = {2019}
}

@article{CAS+22,
    author = {Costa, Pedro C. S. and An, Dong and Sanders, Yuval R. and Su, Yuan and Babbush, Ryan and Berry, Dominic W.},
    title = {Optimal scaling quantum linear-systems solver via discrete adiabatic theorem},
    journal = {PRX Quantum},
    volume = {3},
    number = {4},
    pages = {040303},
    doi = {10.1103/PRXQuantum.3.040303},
    year = {2022}
}

@article{KLL+17,
    author = {Kimmel, Shelby and Lin, Cedric Yen-Yu and Low, Guang Hao and Ozols, Maris and Yoder, Theodore J.},
    title = {Hamiltonian simulation with optimal sample complexity},
    journal = {npj Quantum Information},
    volume = {3},
    number = {1},
    pages = {1--7},
    doi = {10.1038/s41534-017-0013-7},
    year = {2017}
}

@article{LMR14,
    author = {Lloyd, Seth and Mohseni, Masoud and Rebentrost, Patrick},
    journal = {Nature Physics},
    title = {Quantum principal component analysis},
    volume = {10},
    number = {9},
    pages = {631-633},
    doi = {10.1038/nphys3029},
    year = {2014}
}

@inproceedings{ALL22,
    author = {Anshu, Anurag and Landau, Zeph and Liu, Yunchao},
    title = {Distributed quantum inner product estimation},
    booktitle = {Proceedings of the 54th Annual ACM SIGACT Symposium on Theory of Computing},
    pages = {44--51},
    doi = {10.1145/3519935.3519974},
    year = {2022}
}

@article{Hol73,
    author = {Holevo, Alexander S.},
    title = {Statistical decision theory for quantum systems},
    journal = {Journal of Multivariate Analysis},
    volume = {3},
    number = {4},
    pages = {337--394},
    doi = {10.1016/0047-259X(73)90028-6},
    year = {1973}
}

@article{Hel67,
    author = {Helstrom, Carl W.},
    title = {Detection theory and quantum mechanics},
    journal = {Information and Control},
    volume = {10},
    number = {3},
    pages = {254--291},
    doi = {10.1016/S0019-9958(67)90302-6},
    year = {1967}
}

@book{NC10,
    author = {Nielsen, Michael A. and Chuang, Isaac L.},
    title = {Quantum Computation and Quantum Information},
    publisher = {Cambridge University Press},
    doi = {10.1017/CBO9780511976667},
    year = {2010}
}

@article{BCWdW01,
    author = {Buhrman, Harry and Cleve, Richard and Watrous, John and de Wolf, Ronald},
    title = {Quantum Fingerprinting},
    journal = {Physical Review Letters},
    volume = {87},
    number = {16},
    pages = {167902},
    doi = {10.1103/PhysRevLett.87.167902},
    year = {2001}
}

@article{BBD+97,
    author = {Barenco, Adriano and Berthiaume, Andr\'{e} and Deutsch, David and Ekert, Artur and Jozsa, Richard and Macchiavello, Chiara},
    title = {Stabilization of Quantum Computations by Symmetrization},
    journal = {SIAM Journal on Computing},
    volume = {26},
    number = {5},
    pages = {1541--1557},
    doi = {10.1137/S0097539796302452},
    year = {1997}
}

@article{KMY09,
    author = {Kobayashi, Hirotada and Matsumoto, Keiji and Yamakami, Tomoyuki},
    title = {Quantum {Merlin-Arthur} proof systems: are multiple {Merlins} more helpful to {Arthur}?},
    journal = {Chicago Journal of Theoretical Computer Science},
    volume = {2009},
    number = {},
    pages = {3},
    doi = {10.4086/cjtcs.2009.003},
    year = {2009}
}

@article{VV17,
    author = {Valiant, Gregory and Valiant, Paul},
    title = {Estimating the unseen: improved estimators for entropy and other properties},
    journal = {Journal of the ACM},
    volume = {64},
    number = {6},
    pages = {37:1--37:41},
    doi = {10.1145/3125643},
    year = {2017}
}

@article{Che00,
    author = {Chefles, Anthony},
    title = {Quantum state discrimination},
    journal = {Contemporary Physics},
    volume = {41},
    number = {6},
    pages = {401--424},
    doi = {10.1080/00107510010002599},
    year = {2000}
}

@article{BC09,
    author = {Barnett, Stephen M. and Croke, Sarah},
    title = {Quantum state discrimination},
    journal = {Advances in Optics and Photonics},
    volume = {1},
    number = {2},
    pages = {238--278},
    doi = {10.1364/AOP.1.000238},
    year = {2009}
}

@article{BK15,
    author = {Bae, Joonwoo and Kwek, Leong-Chuan},
    title = {Quantum state discrimination and its applications},
    journal = {Journal of Physics A: Mathematical and Theoretical},
    volume = {48},
    number = {8},
    pages = {083001},
    doi = {10.1088/1751-8113/48/8/083001},
    year = {2015}
}

@inproceedings{BOW19,
    author = {B{\u{a}}descu, Costin and O'Donnell, Ryan and Wright, John},
    title = {Quantum state certification},
    booktitle = {Proceedings of the 51st Annual ACM SIGACT Symposium on Theory of Computing},
    pages = {503--514},
    doi = {10.1145/3313276.3316344},
    year = {2019}
}

@article{HHJ+17,
    author = {Haah, Jeongwan and Harrow, Aram W. and Ji, Zhengfeng and Wu, Xiaodi and Yu, Nengkun},
    title = {Sample-optimal tomography of quantum states},
    journal = {IEEE Transactions on Information Theory},
    volume = {63},
    number = {9},
    pages = {5628--5641},
    doi = {10.1109/TIT.2017.2719044},
    year = {2017}
}

@inproceedings{OW16,
    author = {O'Donnell, Ryan and Wright, John},
    booktitle = {Proceedings of the 48th Annual ACM Symposium on Theory of Computing},
    title = {Efficient quantum tomography},
    pages = {899--912},
    doi = {10.1145/2897518.2897544},
    year = {2016}
}

@misc{Kit95,
    author = {Kitaev, A. Yu.},
    title = {Quantum measurements and the {Abelian} stabilizer problem},
    howpublished = {ArXiv e-prints},
    eprint = {quant-ph/9511026},
    year = {1995}
}

@article{Sho97,
    author = {Shor, Peter W.},
    title = {Polynomial-time algorithms for prime factorization and discrete logarithms on a quantum computer},
    journal = {SIAM Journal on Computing},
    volume = {26},
    number = {5},
    pages = {1484--1509},
    doi = {10.1137/S0097539795293172},
    year = {1997}
}

@article{GLF+10,
    author = {Gross, David and Liu, Yi-Kai and Flammia, Steven T. and Becker, Stephen and Eisert, Jens},
    title = {Quantum state tomography via compressed sensing},
    journal = {Physical Review Letters},
    volume = {105},
    number = {15},
    pages = {150401},
    doi = {10.1103/PhysRevLett.105.150401},
    year = {2010}
}

@article{GT09,
    author = {G{\"{u}}hne, Otfried and G{\'{e}}za T{\'{o}}th},
    title = {Entanglement detection},
    journal = {Physics Reports},
    volume = {474},
    number = {1--6},
    pages = {1--75},
    doi = {10.1016/j.physrep.2009.02.004},
    year = {2009}
}

@article{FL11,
    author = {Flammia, Steven T. and Liu, Yi-Kai},
    title = {Direct fidelity estimation from few {Pauli} measurements},
    journal = {Physical Review Letters},
    volume = {106},
    number = {23},
    pages = {230501},
    doi = {10.1103/PhysRevLett.106.230501},
    year = {2011}
}

@article{WZC+23,
    author = {Wang, Qisheng and Zhang, Zhicheng and Chen, Kean and Guan, Ji and Fang, Wang and Liu, Junyi and Ying, Mingsheng},
    title = {Quantum Algorithm for Fidelity Estimation},
    journal = {IEEE Transactions on Information Theory},
    volume = {69},
    number = {1},
    pages = {273--282},
    doi = {10.1109/TIT.2022.3203985},
    year = {2023}
}

@article{WGL+24,
    author = {Wang, Qisheng and Guan, Ji and Liu, Junyi and Zhang, Zhicheng and Ying, Mingsheng},
    title = {New Quantum Algorithms for Computing Quantum Entropies and Distances},
    journal = {IEEE Transactions on Information Theory},
    volume = {70},
    number = {8},
    pages = {5653--5680},
    doi = {10.1109/TIT.2024.3399014},
    year = {2024}
}

@misc{GP22,
    author = {{Gily\'{e}n}, Andr\'{a}s and Poremba, Alexander},
    title = {Improved quantum algorithms for fidelity estimation},
    eprint = {2203.15993},
    howpublished = {ArXiv e-prints},
    year = {2022}
}

@article{WZ24b,
    author = {Wang, Qisheng and Zhang, Zhicheng},
    title = {Fast quantum algorithms for trace distance estimation},
    journal = {IEEE Transactions on Information Theory},
    volume = {70},
    number = {4},
    pages = {2720--2733},
    doi = {10.1109/TIT.2023.3321121},
    year = {2024}
}

@article{LGLW23,
    author = {Le Gall, Fran{\c{c}}ois and Liu, Yupan and Wang, Qisheng},
    title = {Space-bounded quantum state testing via space-efficient quantum singular value transformation},
    journal = {Computational Complexity},
    volume = {},
    number = {},
    pages = {},
    doi = {},
    year = {2026},
    eprint = {2308.05079},
}

@article{TYKI06,
    author = {Tokunaga, Yuuki and Yamamoto, Takashi and Koashi, Masato and Imoto, Nobuyuki},
    title = {Fidelity estimation and entanglement verification for experimentally produced four-qubit cluster states},
    journal = {Physical Review A},
    volume = {74},
    number = {2},
    pages = {020301(R)},
    doi = {10.1103/PhysRevA.74.020301},
    year = {2006}
}

@article{GLGP07,
    author = {G{\"{u}}hne, Otfried and Lu, Chao-Yang and Gao, Wei-Bo and Pan, Jian-Wei},
    title = {Toolbox for entanglement detection and fidelity estimation},
    journal = {Physical Review A},
    volume = {76},
    number = {3},
    pages = {030305(R)},
    doi = {10.1103/PhysRevA.76.030305},
    year = {2007}
}

@article{OW21,
    author = {O'Donnell, Ryan and Wright, John},
    title = {Quantum spectrum testing},
    journal = {Communications in Mathematical Physics},
    volume = {387},
    number = {1},
    pages = {1--75},
    doi = {10.1007/s00220-021-04180-1},
    year = {2021}
}

@article{Wan24,
  title={Optimal trace distance and fidelity estimations for pure quantum states},
  author={Wang, Qisheng},
  journal={IEEE Transactions on Information Theory},
  volume={70},
  number={12},
  pages={8791--8805},
  year={2024},
  doi={10.1109/TIT.2024.3447915}
}

@incollection{BHMT02,
    author = {Brassard, Gilles and H{\o}yer, Peter and Mosca, Michele and Tapp, Alain},
    title = {Quantum amplitude amplification and estimation},
    editor = {Lomonaco, Jr., Samuel J. and Brandt, Howard E.},
    booktitle = {Quantum Computation and Information},
    volume = {305},
    number = {},
    pages = {53--74},
    doi = {10.1090/conm/305/05215},
    publisher = {AMS},
    series = {Contemporary Mathematics},
    year = {2002}
}

@article{BBC+01,
    author = {Beals, Robert and Buhrman, Harry and Cleve, Richard and Mosca, Michele and de Wolf, Ronald},
    title = {Quantum lower bounds by polynomials},
    journal = {Journal of the ACM},
    volume = {48},
    number = {4},
    pages = {778--797},
    doi = {10.1145/502090.502097},
    year = {2001}
}

@inproceedings{NW99,
    author = {Nayak, Ashwin and Wu, Felix},
    title = {The quantum query complexity of approximating the median and related statistics},
    booktitle = {Proceedings of the 31st Annual ACM Symposium on Theory of Computing},
    pages = {384--393},
    doi = {10.1145/301250.301349},
    year = {1999}
}

@inproceedings{AS24,
    author = {Arunachalam, Srinivasan and Schatzki, Louis},
    title = {Generalized inner product estimation with limited quantum communication},
    booktitle = {Proceedings of the 42nd International Symposium on Theoretical Aspects of Computer Science},
    pages = {11:1--11:17},
    doi = {10.4230/LIPIcs.STACS.2025.11},
    year = {2025}
}

@misc{GHYZ24,
    author = {Gong, Weiyuan and Haferkamp, Jonas and Ye, Qi and Zhang, Zhihan},
    title = {On the sample complexity of purity and inner product estimation},
    eprint = {2410.12712},
    howpublished = {ArXiv e-prints},
    year = {2024}
}

@article{HKP20,
  title={Predicting many properties of a quantum system from very few measurements},
  author={Huang, Hsin-Yuan and Kueng, Richard and Preskill, John},
  journal={Nature Physics},
  volume={16},
  number={10},
  doi={10.1038/s41567-020-0932-7},
  pages={1050--1057},
  year={2020}
}

@inproceedings{HH00,
    author = {Hales, Lisa and Hallgren, Sean},
    title = {An improved quantum Fourier transform algorithm and applications},
    booktitle = {Proceedings of the 41st Annual Symposium on Foundations of Computer Science},
    pages = {515--525},
    doi = {10.1109/SFCS.2000.892139},
    year = {2000}
}

@misc{GKP+24,
    author = {Go, Byeongseon and Kwon, Hyukjoon and Park, Siheon and Patel, Dhrumil and Wilde, Mark M.},
    title = {Density matrix exponentiation and sample-based {Hamiltonian} simulation: Non-asymptotic analysis of sample complexity},
    howpublished = {ArXiv e-prints},
    eprint = {2412.02134},
    year = {2024}
}

@article{LWWZ24,
    author = {Liu, Nana and Wang, Qisheng and Wilde, Mark M. and Zhang, Zhicheng},
    title = {Quantum algorithms for matrix geometric means},
    journal = {npj Quantum Information},
    volume = {11},
    number = {},
    pages = {101},
    doi = {10.1038/s41534-025-00973-7},
    year = {2025}
}

    \appendix

    \section{Relation between Pure-State and General Samplizers} \label{sec:general-samplizer}

    The notion of block-encoding oracle is used in the general samplizer.

    \begin{definition} [Block-encoding oracle]
        Let $A$ be an operator with $\Abs{A} \leq 1$. 
        A unitary operator $U$ is said to be a block-encoding oracle for $A$, if $U$ is a $\rbra{2, a, 0}$-block-encoding of $A$ for some $a \geq 1$. 
    \end{definition}
    The block-encoding oracle is now a common quantum input model employed in various computational problems, e.g., Hamiltonian simulation \cite{GSLW19} and solving systems of linear equations \cite{CAS+22}. 
    It is known that for pure states, the block-encoding oracle is equivalent to the reflection oracle in query complexity. 

    Next, we recall the definition of the general samplizer. 

    \begin{definition} [General samplizer, {\cite[Definition I.1]{WZ24}}] \label{def:samplizer-general}
        A general samplizer $\mathsf{Samplize}_*\ave{*}$ is a converter from a quantum circuit family to a quantum channel family with the following property. 
        For any $\delta > 0$, quantum circuit family $\mathcal{A}^U$ with query access to unitary operator $U$, and mixed state $\rho$, there exists a block-encoding oracle $U_\rho$ for $\rho$ such that
        \begin{equation}
            \Abs{\mathsf{Samplize}_\delta\ave{\mathcal{A}^U}\sbra{\rho} - \mathcal{A}^{U_\rho}}_\diamond \leq \delta.
        \end{equation}
        The sample complexity of $\mathsf{Samplize}_*\ave{*}$ is a function $S\rbra{\cdot, \cdot}$ such that if $\mathcal{A}^U$ uses $Q$ queries to $U$, then $\mathsf{Samplize}_\delta\ave{\mathcal{A}^U}\sbra{\rho}$ uses $S\rbra{Q, \delta}$ samples of $\rho$. 
    \end{definition}

    In \cref{def:samplizer-pure}, we define the (multi-)samplizer for pure states in terms of the reflection oracle, while in \cref{def:samplizer-general}, the general samplizer is defined in terms of the block-encoding oracle. 
    This is for convenience of use and without loss of generality.
    Note that for pure states, the block-encoding oracle turns out to be equivalent to the reflection oracle (while we are not aware of any counterpart for mixed states). 

    \begin{fact} [cf.\ {\cite[Lemma 5.5]{CWZ24}}] \label{fact:equiv-reflect-block-encoding}
        The reflection oracle for $\ket{\psi}$ and the block-encoding oracle for $\ketbra{\psi}{\psi}$ are equivalent in query complexity up to a (universal) constant factor.
    \end{fact}
    Therefore, the notion of $1$-samplizer for pure states is a special case of the general samplizer defined in \cite{WZ24}.
    In the following, we show how $1$-samplizer can be efficiently implemented by general samplizer.

    \begin{proposition} \label{prop:general-samplizer-to-1-samplizer}
        A general samplizer with sample complexity $S\rbra{Q, \delta}$ implies a $1$-samplizer for pure states with sample complexity $S\rbra{cQ, \delta}$ for some universal constant $c \geq 1$. 
    \end{proposition}
    \begin{proof}
        Let $\mathsf{Samplize}_*\ave{*}$ be a general samplizer with sample complexity $S\rbra{Q, \delta}$. 
        Let $\ket{\psi}$ be an $n$-qubit pure state and $\mathcal{A}^{V}$ be a quantum query algorithm using $Q$ queries to an $n$-qubit unitary oracle $V$.
        According to \cref{fact:equiv-reflect-block-encoding}, there is a quantum query algorithm $\mathcal{B}^{U}$ using $cQ$ queries to $U$ for some universal constant $c \geq 1$ such that for any block-encoding oracle $U_\psi$ for $\ketbra{\psi}{\psi}$, $\mathcal{A}^{R_{\psi}} = \mathcal{B}^{U_\psi}$.
        By applying the general samplizer $\mathsf{Samplize}_*\ave{*}$, for any $\delta > 0$, there exists a block-encoding oracle $U_\psi$ for $\ketbra{\psi}{\psi}$ such that
        \begin{equation}
            \Abs*{\mathsf{Samplize}_\delta\ave{\mathcal{B}^U}\sbra{\ketbra{\psi}{\psi}} - \mathcal{B}^{U_\psi}}_\diamond \leq \delta.
        \end{equation}
        Here, we note that $\mathsf{Samplize}_\delta\ave{\mathcal{B}^U}\sbra{\ketbra{\psi}{\psi}}$ uses $S\rbra{cQ, \delta}$ samples of $\ket{\psi}$. 
        Let $f \colon \mathcal{A}^{V} \mapsto \mathcal{B}^{U}$ be the mapping that is independent of $\ket{\psi}$. 
        Then,
        \begin{equation}
            \Abs*{\mathsf{Samplize}_\delta\ave{f\rbra{\mathcal{A}^{V}}}\sbra{\ketbra{\psi}{\psi}} - \mathcal{A}^{R_{\psi}}}_\diamond \leq \delta.
        \end{equation}
        Therefore, a $1$-samplizer can be formally given by
        \begin{equation}
            \mathsf{Samplize}^{\mathsf{pure}}_\delta\ave{\cdot}\sbra{\ket{\psi}} \coloneqq \mathsf{Samplize}_\delta\ave{f\rbra{\cdot}}\sbra{\ketbra{\psi}{\psi}},
        \end{equation}
        and its sample complexity is $S\rbra{cQ, \delta}$. 
    \end{proof}

    A general samplizer was proposed in \cite[Theorem III.1]{WZ24} with sample complexity $S\rbra{Q, \delta} = O\rbra*{\frac{Q^2}{\delta}\log^2\rbra*{\frac{Q}{\delta}}}$. 
    By \cref{prop:general-samplizer-to-1-samplizer}, we immediately obtain a $1$-samplizer for pure states with sample complexity $S\rbra{cQ, \delta} = O\rbra*{\frac{Q^2}{\delta}\log^2\rbra*{\frac{Q}{\delta}}}$. 
    In \cref{sec:approach-samplizer}, we provided a $1$-samplizer for pure states with better sample complexity, removing the logarithmic factors, and generalize it to multi-samplizer for pure states.

\end{document}